\documentclass[12pt]{article}
\usepackage{epsfig}
\usepackage{epic}
\usepackage{color}

\usepackage{graphicx}
\usepackage{makeidx}
\usepackage{setspace}
\usepackage{amssymb,amsmath,amsthm}
\newtheorem{theorem}{Theorem}{}
\newtheorem{corollary}{Corollary}{}
{}
{}
{}
\newtheorem{lemma}{Lemma}{}
\newtheorem{remark}{Remark}{}
{}
{}
{}
{}
{}

\newcommand{\code}{\mathcal{C}}

\newcommand{\prob}{\mathbb{P}}

\newcommand{\openone}{\leavevmode\hbox{\small1\normalsize\kern-.33em1}}
\newcommand {\bE} {\mathbb{E}}

\newcommand {\beq} {\begin{equation}}
\newcommand {\eeq} {\end{equation}}

\newcommand {\noise} {l^n}

\newcommand {\bad} {\mathcal{B}}
\newcommand {\ueta} {\widehat{\eta}}

\newcommand {\dlm} {\texttt{l}_{\text{max}}}

\newcommand {\dcode} {{\mathcal C}^{\perp}}

\newcommand{\dist}[2]{\text{dist} (#1, #2)}

\newcommand {\spins} {{\mathfrak V}}

\begin{document}
\title{\sc Decay of Correlations for Sparse Graph Error Correcting Codes}

%

\author{Shrinivas Kudekar, Nicolas Macris \\
{\small Communication Theory Laboratory}\\
{\small School of Computer and Communication Sciences}\\ 
{\small Ecole Polytechnique F\'ed\'erale Lausanne}\\
{\small CH-1015 Lausanne, Switzerland}\\
{\small shrinivas.kudekar@epfl.ch, nicolas.macris@epfl.ch}}

\maketitle

\begin{abstract} 
The subject of this paper is transmission over a general class of binary-input memoryless symmetric channels using error correcting codes based on sparse graphs, namely low-density generator-matrix and low-density parity-check codes.
The optimal (or ideal) decoder based on the posterior measure over the code bits, and its relationship to the sub-optimal belief propagation decoder, are investigated. We consider the correlation (or covariance) between two
codebits, averaged over the noise realizations, as a function of the graph distance, for the optimal decoder.  Our main result is that this correlation decays exponentially fast for fixed general low-density generator-matrix codes and high enough noise parameter, and also for fixed general
low-density parity-check codes and low enough noise parameter. This has many consequences. Appropriate performance curves - called GEXIT functions - of the belief propagation and optimal decoders match in high/low  noise regimes. This means that in high/low noise regimes the performance curves of the optimal decoder can be computed by density evolution. Another interpretation is that the  replica predictions of spin-glass theory are exact. Our methods are rather general and use cluster expansions first developed in the context of mathematical statistical mechanics.
\end{abstract}

\section{\sc Introduction}\label{section-1}

Low-density parity-check (LDPC) codes based on sparse graphs have emerged as a
focal point in the theory of error correcting codes, used in noisy channel
communication, largely because they are amenable to low complexity decoding and
at the same time have a good performance (measured as the gap to Shannon's
capacity). An important class of low complexity decoders are the message
passing iterative decoders. In this framework, in order to decode a bit
attached to a node of the graph, one unravels a computational tree (or covering
tree) and iteratively updates messages (suitable functions of the channel
output observations) passed along the edges of the computational tree. We refer
to the recent book \cite{Ruediger} for the state of the art of this general
theory.  One would also like to be able to compare sub-optimal message passing
decoders with the optimal or ideal decoder. The later is based on the posterior
probability distribution supported on code-bits and is optimal in the sense
that it is known to minimize the bit-error-rate among all decoders (it is also
called MAP decoder, and this is the terminology that we adopt in this paper). A
priori the comparison of decoders is not easily done since the MAP decoder is
in general computationaly complex.

One of the most important low complexity message passing decoders is the belief
propagation (BP) algorithm. It is well known that for a code whose graph is a
tree, the BP algorithm has the same performance as the MAP decoder. This
essentially comes from the fact that on a tree the computational graph of a
node matches the original graph itself.  However, codes based on tree graphs
have poor performance and one needs to consider graphs with loops or cycles.
With cycles in the original graph, the messages on the computational tree are
no longer independent, and it is not a priori clear, if and why, the BP
algorithm should retain any close relationship to the MAP decoder.  A
fundamental theoretical tool that allows to analyze the BP algorithm is density
evolution (DE) first developed in \cite{RU}. From DE one can for example
obtain a noise threshold above which reliable communication is not possible
with BP decoding.  The analysis proceeds by taking first a very large block
length $n$ and looking at $d\ll n$ iterations of the BP decoder. Eventualy one
considers the asymptotics $\lim_{d\to+\infty}\lim_{n\to+\infty}$. However, in
the practical  use of the decoder one fixes $n$ (large) and the iterations
$d\gg n$ are performed till one reaches an acceptably small bit-error-rate.
This corresponds to the asymptotics $\lim_{n\to+\infty}\limsup/{\rm
inf}_{d\to+\infty}$. The practical success of density evolution relies on the
equivalence of these two limiting procedures, an open problem in general. 

It is fair to say that these issues have been resolved over the binary erasure
channel (BEC) \cite{Ruediger} (the analysis however has not proceeded
from the point of view of the correlations of the MAP decoder).  An important
tool in the analysis over the BEC has been the "extended BP`` extrinsic
information transfer (EXIT) curve (which is a suitable continuation of the
bit-error-rate curve under BP decoding). Recently it was shown that the
bit-error-rate of the MAP decoder can be obtained from that of the extended BP
EXIT curve by a Maxwell construction just as in the theory of first order phase
transitions\footnote{The extended BP curve corresponds to the pressure-volume
curve of the Van der Waals theory of the liquid-gas transition, and the MAP
curve corresponds to the isotherms obtained by Maxwell's equal area
construction.} \cite{CyRuMon}, \cite{CyRuMonRi}. This construction allows to
compute the MAP noise threshold, and to compare it to the BP noise threshold.
The validity of the exchange of limits $d\to +\infty$ and $n\to +\infty$ (for
the BP decoder) can also be derived for the BEC using natural monotonicity
properties of the decoder \cite{Ruediger}.

For the case of transmission over more general channels very little is known
about these issues. Indeed there one lacks the combinatorial methods available
for the BEC, and radically new methods have to be used. Convenient measures of
the performance, which generalize the EXIT curves, are the so-called GEXIT
curves \cite{Ruediger} (see the next section for their precise definition). It is
believed that in terms of these, the results obtained for the BEC still hold.
In particular, the GEXIT curves for the BP and the MAP decoder should match for
high and low noise regimes away from the phase transition thresholds. Such
conjectures are supported by spin glass theory calculations (e.g the replica
and cavity methods) which provide conjectural but analytic formulas. One-sided
bounds have been derived for the GEXIT curves by the (information theoretical)
method of physical degradation \cite{Ruediger} and also by using correlation
inequalities valid for spin glasses \cite{Macris2}. Related bounds on the conditional
input-output entropy have also been derived \cite{Montanari}, \cite{KuM} by using
"interpolation methods" first developed in the mathematical theory of spin
glasses \cite{GuerraTon}, \cite{Franz}, \cite{Talagrand-book}.  As it turns out, all
these bounds match the replica expressions and are therefore believed to be the
best possible.  In \cite{KoKuMa07} the interpolation method has been extended
to obtain the converse bounds for a class of Poissonian LDPC codes over the
BEC, thus recovering combinatorial results of \cite{CyRuMon07isit}
in a completely different way. Concerning the problem of exchanging the $d,
n\to+\infty$ limits we refer to \cite{KoRu} for recent progress that goes beyond
the BEC.

In this work we will show that a good deal can be learned by looking at the
correlations (more precisely the covariance), averaged over the channel
outputs, of the MAP decoder. We comment below about the methods used, but let
us say at the outset that our aim is to cover a fairly general class of binary
input memoryless symmetric channels, including the binary symmetric and
gaussian ones.  One of our main results is that for sufficiently low noise
(LDPC codes) the correlations between two code-bits decay exponentially fast as
a function of the graph distance between the two code-bits, uniformly in the
block length size $n$.  The sparsity of the underlying graph then implies that,
if furthermore the decay rate beats the local expansion of the graph, the MAP
GEXIT curve can be computed by DE. Another interpretation of this result is
that the solutions provided by the replica/cavity methods of spin glass theory
are exact.

Low-density generator-matrix codes (LDGM) codes have a very clear relationship
to spin glass models on random graphs and it is useful to study them before we
can attack the harder case of LDPC codes. Besides, the present analysis could
potentially be useful in other contexts where they are used (e.g. rateless
codes, source coding).  For high noise we
prove the decay of correlations and that the MAP GEXIT curve can be computed by
DE. For that system, we can also show that the decay of correlations implies
that the limits $d,n\to+\infty$ can be exchanged for the BP decoder, at least
on the binary symmetric channel.

The study of the behavior of correlations as a function of the distance between
local degrees of freedom is one of the central aims of statistical physics.
For lattice spin systems (e.g. the Ising model) an important criterion that
ensures correlation decay is Dobrushin's criterion \cite{Georgii} - which is of
probabilistic nature - and its various improvements. The main
other method - which is not necessarily of probabilistic nature - is based on
suitable expansions in powers of ``the strength of interactions''. There exists a
host of such expansions collectively called ``cluster expansions'', and the
context of spin systems the first and simplest such expansion is the so-called
``polymer expansion'' \cite{Brydges}. The main rule of thumb is that all these methods
work if the degrees of freedom are weakly interacting, or if one can transform
the original system into an effective one involving new weakly interacting
degrees of freedom. It turns out that sophisticated forms of the cluster
expansions can be carried out, for LDGM codes in a high noise regime, and for
LDPC codes in a low noise regime, for a fairly general class of channels
including the binary symmetric channel (BSC) as well as the binary input
additive white gaussian noise channel (BIAWGNC). As we will explain later it is
necessary to use quite sophisticated cluster expansions for at least two
reasons. Concerning LDGM codes Dobrushin's criterion and the polymer expansion
require bounded channel outputs (and thus do not covers the case of the
BIAWGNC). Concerning LDPC codes one has to transform the system to a dual one
that involves "negative Gibbs weights" and cannot even be treated by
probabilistic methods.

The rest of the paper is organized as follows. In section \ref{section-2} we
formulate the models and give a unified view of the main results both for LDGM
and LDPC codes. The main strategy of the proofs is also explained there.
Sections \ref{section-3} and \ref{section-4} contain the proofs of correlation
decay and its consequences for the GEXIT curves. The problem of exchanging the
limits of iteration number and block length size is addressed in section
\ref{section-5}. We conclude by pointing out open problems and further
connections to the recent literature. The appendix reviews in a streamlined
form the two cluster expansions that are used in sections \ref{section-3} and
\ref{section-4}.

Summaries of the present results have been reported for the special case of the BIAWGNC \cite{KuMa08itw}, \cite{KuMa09isit}.

\section{\sc Models and Main Results}\label{section-2}

We consider binary-input memoryless output-symmetric channels defined by a transition p.d.f $p_{Y\mid X}(y\mid x)$ with inputs $x\in\{-1,+1\}$ and outputs belonging to $\mathbb{R}$. Since we use techniques from statistical mechanics it is convenient to immediately map the usual input alphabet $\{0,1\}$ to $\{-1, +1\}$. Symmetry of the channel means that $p_{Y\mid X}(-y\mid -x) = p_{Y\mid X}(y\mid x)$. The intensity of the noise is called $\epsilon$. 
It will be convenient to trade off the channel outputs $y$ for the half-loglikelihood
\begin{equation}
l = \frac{1}{2}\ln \biggl[\frac{p_{Y\vert X}(y\vert +1)}{p_{Y\vert X}(y\vert -1)}\biggr]
\end{equation} 
It is well known that on a symmetric channel one can assume without loss of generality that the all-one codeword (i.e the usual all-zero codeword) is transmitted and therefore the channel outputs are i.i.d with distribution $p_{Y\mid X}(y\vert +1)dy\equiv c(l) dl$.
For clarity, we assume that the noise parameter varies in an
interval $[0,\epsilon_{max}]$ ($\epsilon_{max}$ is possibly infinite) where
$\epsilon\to 0$ corresponds to low noise and $\epsilon\to \epsilon_{max}$
corresponds to high noise. For example, $\epsilon_{max}=\frac12$ for the BSC,
$\epsilon_{max}=+\infty$ for the BIAWGNC (and $\epsilon_{max}=1$ for the BEC). 
The general class of channels for which our main results hold is\\

\noindent{\bf Class of channels.} {\it We define the class $\mathcal{K}$ of binary-input memoryless output-symmetric channels:
\begin{enumerate}
\item 
The numbers $T_{2p}(\epsilon)= \frac{d}{d\epsilon}\int_{-\infty}^{\infty} dl\, c(l) (\tanh l)^{2p}$ are bounded uniformly with respect to $p\geq 1$ integer.
\item 
For any finite $m>0$ we have $\mathbb{E}[e^{m\vert l\vert}]\leq c_m<+\infty$
\item{(Low noise condition)}
There exists $s_0>0$ small enough such that for $0<s\leq s_0$ we have $\lim_{\epsilon\to 0}\mathbb{E}[e^{-sl}] = 0$.
\item{(High noise condition)} Set $\delta(\epsilon, H) = e^{4H}-1 + \mathbb{P}(|l| >
H)$. One can find $H(\epsilon)$ such that $\lim_{\epsilon \to \epsilon_{max}}
\delta(\epsilon, H(\epsilon))=0$.
\end{enumerate}
}

\vskip 0.25cm
Note that this class is not the most general that we can treat but it is at the same time fairly general and keeps the analysis at a technically reasonable level. An important example is the BSC (we keep $0\leq \epsilon\leq \frac{1}{2})$
\begin{align}\label{bsc}
 p_{Y\vert X}(y\vert x)& = (1-\epsilon)\delta(y-x) +\epsilon\delta(y + x), \nonumber \\
 c(l) & = (1-\epsilon)\delta(l- \frac{1}{2}\ln \frac{1-\epsilon}{\epsilon}) + \epsilon\delta((l- \frac{1}{2}\ln \frac{\epsilon}{1-\epsilon})
\end{align}
One can check that the conditions are met with $T_{2p}(\epsilon)= 2p(1-2\epsilon)^{2p-1}$, $\mathbb{E}[e^{m\vert l\vert}] = (\frac{1-\epsilon}{\epsilon})^{\frac{m}{2}}$, $\mathbb{E}[e^{-sl}] = \epsilon^{\frac{s}{2}}(1-\epsilon)^{1-\frac{s}{2}} + (1-\epsilon)^{\frac{s}{2}}\epsilon^{1-\frac{s}{2}}$ and $H(\epsilon) = \log \frac{1-\epsilon}{\epsilon}$. Another important example is the BIAWGNC
\begin{equation}\label{gaussian}
 p_{Y\vert X}(y\vert x) = \frac{1}{\sqrt{2\pi}\epsilon}
\exp\biggl(-\frac{(y-x)^2}{2\epsilon}\biggr),\,\, c(l)= {\frac{\epsilon}{\sqrt{2 \pi}}}\exp\biggl(-\frac{(l-\epsilon^{-1})^2}{2\epsilon^{-1}}\biggr)
\end{equation}
Again one can check that the conditions are met with $T_{2p}(\epsilon) \leq \int_{-\infty}^{+\infty} dl\, \vert\frac{dc(l)}{d\epsilon}\vert$, $\mathbb{E}[e^{m\vert l\vert}] < \infty$, $\mathbb{E}[e^{-sl}] = e^{-s\epsilon^{-1}(1-\frac{s}{2})}$ and $H(\epsilon)= 2 \epsilon^{-1/4}$.
Note that the BEC is not contained in the class $\mathcal{K}$ because of the second condition. Nevertheless due to the special nature of this channel our methods can easily be adapted, but we will not give the details here since this is a case that has already been thoroughly analyzed in the literature \cite{Ruediger}. 

Fixed LDGM codes are constructed from a fixed bipartite graph with $m$ information-bit nodes (variable nodes) and $n$ code-bit nodes (check nodes), and edges connecting variable and check nodes only. The design rate of the code $R=\frac{m}{n}$ is kept fixed. The set of neighbors of a variable node $a$ is called $\partial a$ and the set of neighbors of a check node $i$ is called 
$\partial i$. We consider graphs with bounded node degrees
$\vert \partial a\vert\leq l_{\max}$ and  
$\vert \partial i\vert\leq k_{\max}$. Information bits 
$u_1,...,u_{m}\in \{-1, +1\}^m$ are attached to the variable nodes and the code-bits $x_1,...,x_n$ attached to the check nodes are obtained as 
\begin{equation}
x_i=\prod_{a\in \partial i} u_a,\qquad i=1,...,n
\end{equation}
We also consider ensembles of such codes defined by random graph constructions. We do not explain the details of these constructions here except for saying that an LDGM($\Lambda, P$) ensemble is specified by the generating functions of variable (resp. check) node degree distributions $\Lambda(z)=\sum_{l=1}^{l_{\max}} \Lambda_l z^l$ (resp. $P(z)=\sum_{r=1}^{r_{\max}} P_r z^r$) \cite{Ruediger}.

Fixed LDPC codes are similarly constructed from a fixed bipartite graph with $n$ variable nodes $i=1,...,n$ (this time these are the code-bit nodes) and $m$ check nodes $c=1,...m$, with edges connecting variable and check nodes only. The design rate is $R=1-\frac{m}{n}$ is fixed. We assume that the node degrees are bounded $\vert \partial i \vert\leq l_{\max}$ and $\vert \partial c\vert\leq k_{\max}$. The code-bits $x_1,...,x_{n}$ attached to the variable nodes satisfy $m$ parity check constraints
\begin{equation} 
\prod_{i\in \partial c} x_i =1, \qquad c=1,...,m
\end{equation}
We also consider ensembles of such codes defined by random graph constructions; an ensemble is specified by the generating functions of variable (resp. check) node degree distribution $\Lambda(z)=\sum_{l=1}^{l_{\max}} \Lambda_l z^l$ (resp. $P(z)=\sum_{r=1}^{r_{\max}} P_r z^r$) \cite{Ruediger}.

The optimal MAP decoder is based on the posterior measure of the transmitted codeword given the received message $y^n=(y_1,..,y_n)$. 
For LDGM codes this conditional measure is best viewed as being supported on information bits,
\begin{equation}\label{gibbs-ldgm}
p_{U^m\vert Y^n}(u^m\vert y^n) = \frac{1}{Z} \prod_{i=1}^n e^{l_i\prod_{a\in \partial i} u_a}
\end{equation}
For LDPC codes the conditional measure is
\begin{equation}\label{gibbs-ldpc}
p_{X^n\vert Y^n}(x^n\vert y^n) = \frac{1}{Z} \prod_{c=1}^m \frac{1}{2}(1+\prod_{i\in\partial c} x_i)\prod_{i=1}^ne^{l_ix_i}
\end{equation}
In both cases $Z$ is the appropriate normalization factor.
These measures are random because of the channel outputs and possibly because the code is chosen at random from an ensemble. 
The average with respect to the channel outputs is often denoted by 
$\mathbb{E}_{l^n}$  and the average with respect to a code ensemble is generically denoted by $\mathbb{E}_\code$. We will also use the notation 
$\mathbb{E}_{l^{n\setminus i}}$ when the average is over all outputs except the $i$-th one.
A crucial point is that the interactions or constraints in these measures are 
local so that they can be analyzed with the tools developed in the theory of Gibbs measures \cite{Georgii}. 
We use the bracket notation $\langle f \rangle =$ 
\begin{equation}
\sum_{u^m} f(u^m) p_{U^m\vert Y^n}(u^m\vert y^n),\qquad \sum_{x^n} f(x^n) p_{X^n\vert Y^n}(x^n\vert y^n)
\end{equation}
for the Gibbs averages of functions $f$. It turns out that even for \eqref{gibbs-ldgm} we will only need to look at averages of functions of the transmitted codebits $x^n$; for example $\langle x_i\rangle =\langle \prod_{i\in a} u_a\rangle$. It is important to remember that the bracket is defined for finite $n$ although we do not write explicitly $\langle -\rangle_n$ to alleviate the notations. 
The average (over noise realizations) Gibbs entropy of the two measures is nothing else than Shannon's input-output conditional entropy $\frac{1}{n}H(U^m\vert Y^n)$, $\frac{1}{n}H(X^n\vert Y^n)$ denoted in both cases by $h_n$. The MAP-GEXIT function is simply defined as the $\epsilon$ derivative of this conditional entropy. When this derivative is performed one finds that the MAP-GEXIT function is a functional of the soft-bit MAP estimate\footnote{the magnetization}
$\langle x_i\rangle$ (or the magnetization). It is much more convenient, in fact, to express it as a functional of the extrinsic estimate $\langle x_i\rangle_0$ computed for $l_i=0$,
\begin{equation}\label{gexit}
 g_n(\epsilon)=\frac{d}{d\epsilon}\mathbb{E}_\code[h_n]
 =\mathbb{E}_\code[\mathcal{G}(\langle x_i\rangle_0)]
\end{equation}
The explicit form of the functionals corresponding to LDPC and LDGM codes is given in sections \ref{section-3} and \ref{section-4} (see also \cite{Ruediger}, \cite{Macris2}). 

Let us now describe the BP decoder from the point of view of Gibbs measures. Given a graph $G$ defining a given LDGM or LDPC code with for $n$ and $m$ fixed (large), we choose a code-bit node $i$ and construct the computational tree $T_{d}(i)$ of depth $d$ (even). This is the universal covering tree truncated at distance $d$ from node $i$. We label the variable/check nodes of this tree with new independent labels denoted $n$. Let  
$\pi : T_{d}(i)\to G$ be the projection from the covering tree to the original graph. A node $ \nu\in T_{d}(i)$ has an image $\pi( \nu)$, and due to the loops in $G$ this projection is a many to one map: one may have $\nu\neq \nu^\prime$, $\pi( \nu)=\pi( \nu^\prime)$.
Now, consider a tree-code defined in the usual way on the 
tree-graph $T_{d}(i)$. One can view the BP decoder for node $x_i$ as a MAP decoder for this tree-code. In other words the BP decoder uses the Gibbs measure on $T_d(i)$: one crucial point is that for this Gibbs measure the half-loglikelihood variables attached to the nodes are no longer independent. For the LDGM case the measure is
\begin{equation}\label{BP-ldgm}
\frac{1}{Z_{T_{d}(i)}} \prod_{ k\in T_{d}(i)}e^{l_\pi( k)\prod_{ a\in \partial  k} u_{ a}}
\end{equation}
while for LDPC case
\begin{equation}\label{BP-ldpc}
\frac{1}{Z_{T_{d}(i)}} \prod_{ c \in T_{d}(i)} \frac{1}{2}(1+\prod_{ k\in\partial  c} x_{k})\prod_{ k \in T_{d}(i)}e^{l_{\pi( k)}x_{ k}}
\end{equation}
where in each case $Z_{T_{d}(i)}$ is the proper normalization factor. 
We call $\langle -\rangle^{BP}_d$ the Gibbs bracket with respect to these measures. The extrinsic BP soft-bit estimate is  $\langle x_i \rangle^{BP}_{0,d}$.
The BP-GEXIT function can be defined\footnote{The definition adopted here is very natural from the point of view of the measures \eqref{BP-ldgm} and \eqref{BP-ldpc}. In \cite{Ruediger} another definition is given that is more natural from the point of view of information theory. It is not difficult to show that they are equivalent as $n\to +\infty$} in terms of the same functional than in \eqref{gexit}
\begin{equation}\label{functional-BP}
g_{n,d}^{BP}(\epsilon)=\mathbb{E}_\code[\mathcal{G}(\langle x_i\rangle^{BP}_{0,d})]
\end{equation}
The soft-bit estimate $\langle x_i \rangle^{BP}_d$ can be computed exactly by summing the spins starting from the leaves of $T_d(i)$ all the way up to the root $i$. This computation is left to the reader and yields the usual message passing BP algorithm.

We are now ready to describe our main results.
The main one concerns the exponential decay of the average correlation between two code-bits $x_i$ and $x_j$ as a function of their graph distance ${\rm dist}(i, j)$, uniformly in the system size $n$. 
\begin{theorem}[Decay of correlations for the MAP decoder]\label{theorem-corr}
Consider communication over channels $\mathcal{K}$. Take a fixed LDGM code at high enough noise $\epsilon_g<\epsilon<\epsilon_{max}$ or a fixed LDPC code at low enough noise $0<\epsilon<\epsilon_p$, where $\epsilon_g, \epsilon_p>0$ depend only on $\l_{\max}, k_{\max}$. Then 
\begin{align}\label{eq:avgcorr}
\bE_{l^n}\bigl[\vert\langle x_i x_j \rangle - \langle x_i \rangle \langle x_j \rangle \vert\bigr]\leq c_1e^{-\frac{{\rm dist}(i,j)}{\xi(\epsilon)}}
\end{align}
where $c_1$ is a finite positive numerical constant and $\xi(\epsilon)$ is a strictly positive constant depending 
only on $\epsilon$, $l_{\max}$ and $k_{\max}$. In both regimes we have that $\xi^{-1}(\epsilon)$ grows with $\epsilon\to 0$ and $\epsilon\to \epsilon_{\max}$.
\end{theorem}
Let us say a few words on the strategy used to prove this theorem. As explained in the introduction, for LDGM at high noise and for channels with bounded loglikelihood variables, \eqref{eq:avgcorr} follows from Dobrushin's criterion
or from the polymer expansion. These however do not work when the likelihood variables are unbounded because, roughly speaking, overlapping polymers involve 
moments $\mathbb{E}[l^m]$ which can spoil the convergence as $m\to+\infty$. More physically, what happens is that even in the high noise regime there always exist with positive probability large portions of the graph that are at low noise (or "low temperature")\footnote{See \cite{Frohlich} for a nice discussion of this point related to the Griffith's singularity in the spin glass context}. We use a very convenient cluster expansion of Dreifus-Klein-Perez \cite{Klein} that overcomes this problem by organizing the expansion over self-avoiding random walks on the graph. Since the walks are self-avoiding the moment problem does not occur and we can treat unbounded loglikelihoods. For LDPC codes the situation is more subtle because of the hard parity-check constraints that give an inherently low temperature flavor to the problem.
From a purely code theoretical point of view it is known that LDPC codes are the dual of LDGM codes. This algebraic duality can be exploited to transform the low noise communication model with LDPC codes to a dual model
which, although not a genuine high noise communication model with LDGM codes, still retains this flavor. In fact this dual model involves "negative Gibbs weights". For this reason the cluster expansion of  \cite{Klein} does not work anymore and we use resort to another one first devised by Berretti \cite{Berretti}. The two cluster expansions have to be adapted to our setting and are therefore reviewed in a somewhat streamlined form in 
Appendix \ref{Appendix-A}.

\begin{remark}
The proof of theorem \ref{theorem-corr} will make it clear that for LDGM codes on channels with bounded loglikelihoods (e.g the BSC) at high noise, the average correlation $\mathbb{E}[\langle x_{ i} x_k \rangle_d^{BP} - 
\langle x_{ i} \rangle_d^{BP}\langle x_k \rangle_d^{BP}]$ between the root node and another one decays exponentialy fast (here $\mathbb{E}$ is over the noise. See \cite{Tatikonda} for related work based on Dobrushin's criterion. The likelihood variables over $T_d(i)$ are not independent anymore so that the unbounded case is even more complicated now and will not be discussed here. 
\end{remark}

Our first corollary says that the MAP-GEXIT function can be computed by the DE analysis in high/low noise regimes. It also shows that the replica expressions computed at the appropriate fixed point are exact.
\begin{corollary}[Density evolution allows to compute MAP]\label{corollary-1}
Consider communication over channels $\mathcal{K}$. For ensembles  LDGM($\Lambda, P$) with high enough noise $\epsilon_g^\prime<\epsilon<\epsilon_{max}$ and LDPC($\Lambda, P$) with low enough noise $0<\epsilon< \epsilon_p^\prime$  we have 
\begin{equation}\label{DE}
\lim_{n\to+\infty}g_n(\epsilon)=\lim_{d\to +\infty}\lim_{n\to +\infty}
g^{BP}_{n,d}(\epsilon)
\end{equation}
Here $\epsilon_g^\prime$ and $\epsilon_p^\prime$ depend only on $l_{\max}$, $k_{\max}$.
\end{corollary}
\noindent 
This result extends to the class of channels
$\mathcal{K}$ those obtained previously on the BEC \cite{Ruediger}, \cite{CyRuMon}.
In the case of LDPC ensembles with a vanishing GEXIT curve for $\epsilon\leq \epsilon_*$ it is known that the result can be more easily obtained by physical degradation \cite{CyRuMonRi} or correlation inequalities \cite{Macris2}\cite{Macris} for $\epsilon\leq \epsilon_*$. However there are ensembles with a GEXIT curve that is non trivial all the way down to $\epsilon\to 0$ (for example the Poisson LDPC ensemble) and for which the theorem is new. Note that it applies whether there is or not a phase transition (e.g a jump discontinuity in the GEXIT curve): so it applies even in situations where the area theorem does not allow to prove \eqref{DE}.
The values obtained for $\epsilon_{p,g}^\prime$ are worse than those $\epsilon_{p,g}$ obtained in theorem \ref{theorem-corr}. This is not surprising in view of the following remarks.
It is expected (and for the BEC in some cases it is known) that the equality 
\eqref{DE} is true as long as the noise parameter does not lie in a window 
around the phase transition threshold where this window is determined by an extended form of the BP-GEXIT curve (an S shaped curve). On the other hand inside the window, close to the phase transition threshold, it is known that \eqref{DE} cannot hold. 
A look at the proof shows that the decay of correlations always implies \eqref{DE} only if this decay is fast enough to beat the expansion of the graph: in other words if $\xi\ln (l_{\max} k_{\max}) \ll 1$. Our estimates allow to control the growth of $\xi^{-1}$ with respect to $\epsilon$ to show that such a regime exists.  
Therefore in a window close to the phase transition threshold, even if the correlations decay, $\xi\ln (l_{\max} k_{\max}) \ll 1$ cannot be valid.

Finally concerning the exchange of limits $d,n\to +\infty$ for the BP algorithm we prove
\begin{theorem}[Exchange of limits]\label{theorem-2}
Consider communication over the BSC. For LDGM($\Lambda, P$) ensembles with bounded degrees with high enough noise $\epsilon_g^{\prime\prime}<\epsilon<\epsilon_{max}$, depending only on $l_{\max}, k_{\max}$, we have 
\begin{equation}\label{limit-exchange}
\lim_{d\to +\infty}\lim_{n\to +\infty} g_{n,d}^{BP}(\epsilon)= \lim_{n\to +\infty}\limsup_{d\to +\infty}
g^{BP}_{n,d}(\epsilon) = \lim_{n\to +\infty}\liminf_{d\to +\infty}
g^{BP}_{n,d}(\epsilon)
\end{equation}
\end{theorem}
The proof is a simple application of the decay of correlations. We present it only for the BSC but it can also be extended to any convex combination of such channels and more generaly as long as $c(l)$ has a bounded support that diminishes as the noise parameter increases. The cases of unbounded support (such as BIAWGNC), or of LDPC codes at low noise, require more work and will not be discussed here. The present result complements the recent work \cite{KoRu} which concerns the bit-error-rate of LDPC codes for other message passing decoders in the regime where the error rate vanishes.

\section{\sc LDGM Codes: High Noise}\label{section-3}
In this section we prove theorem \ref{theorem-corr} and its corollary for  LDGM codes. It is convenient to set $K=l_{\max} k_{\max}$.
\begin{proof}[Proof of Theorem \ref{theorem-corr}, LDGM]
First we define the self-avoiding random walks on which the cluster expansion is based.
A self-avoiding walk $w$ between two variable (information-bit) nodes $a, b$ is
a sequence of variable nodes (denoted $v_1, v_2, \dots, v_{l+1}$) and checks
 (denoted $c_1, c_2, \dots, c_{l}$), $v_1,c_1,v_2,c_2,\dots,c_l,v_{l+1}$
such that $v_1=a, v_{l+1}=b$ and $\{v_{m}, v_{m+1}\} \in \partial c_m$ and $v_m\neq v_n, c_m\neq c_n$ for $m\neq n$. We also say that two variable nodes $a, b$ are connected if and only if there exists a self-avoiding walk from $a$ to $b$. 
Thus on a self-avoiding walk we do not repeat variable and check nodes. From any general walk between $a$ and $b$ we can extract a self-avoiding walk
$w$ between $a$ and $b$ which has all its clauses belonging to the parent walk (this is done by chopping off all the loops of the general walk). The
length $\vert w\vert$ of the walk is the number of variable nodes in it. If $a=b$ then the self-avoiding walk from $a$ to $b$ is the trivial walk $a$. We define the length of such walks to be zero. Let $W_{ab}$ denote the set of all self-avoiding walks between variable nodes $a, b$ and 
$W_{AB}=\cup_{a\in A, b\in B}W_{ab}$ (see figure \ref{figcorr}).
Fix some number $H>0$ (that will depend on $\epsilon$ later on). Denote by $\bad$ the set of all code-bit nodes $i$ (checks), such that $|l_i|>H$.
\begin{figure}
\begin{center}
\includegraphics[width=0.6\textwidth, height=0.3\textwidth]{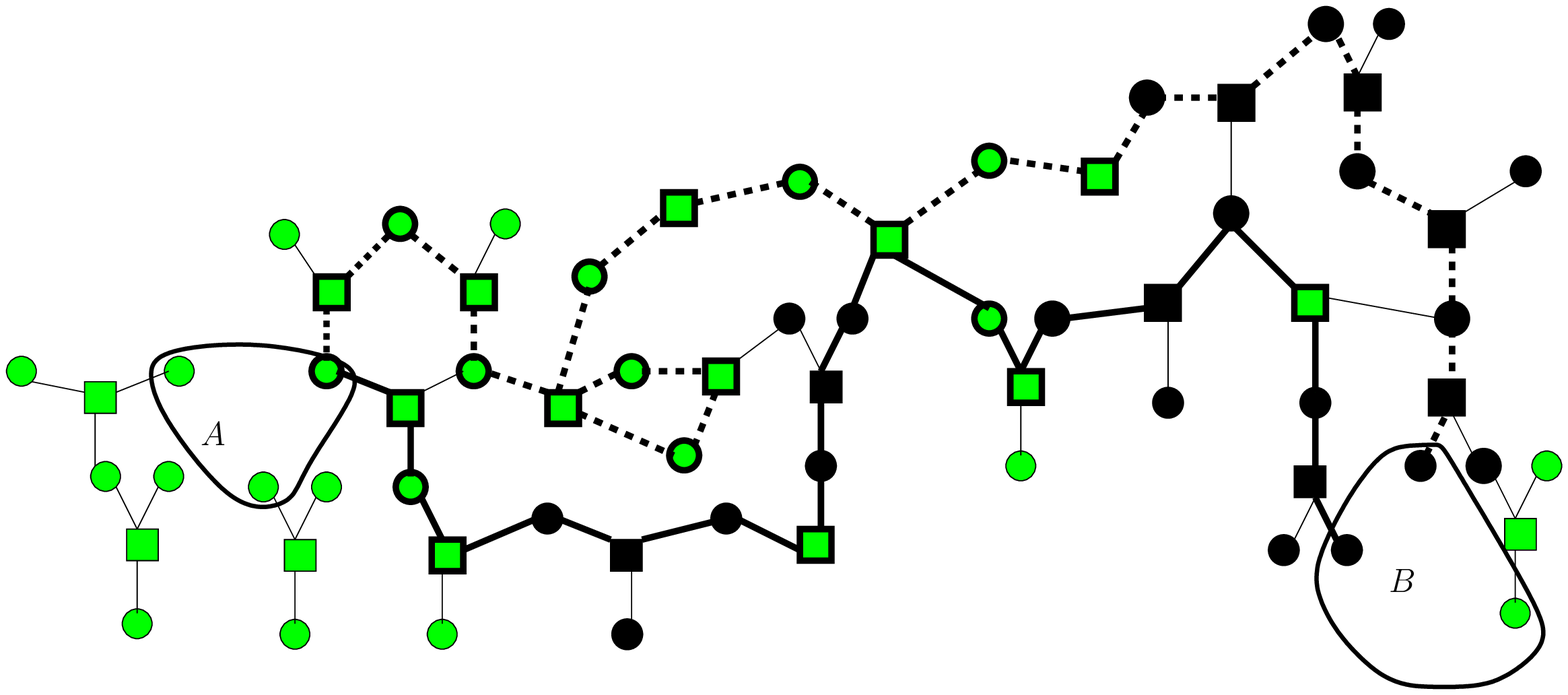}
\vspace{-0.3cm}
\caption{Each set $A$ and $B$ contains three variable nodes. The light squares denote the generator bits in the complement of $\bad$ and the dark squares denote the generator bits in $\bad$. The thick path is an example of a self-avoiding path between $A$ and $B$ which contributes to the upper bound. The dashed path is a non-self-avoiding path and does not contribute to the bound.}\label{figcorr}
\vspace{-0.8cm}
\end{center}
\end{figure}
We use the following (see Appendix \ref{Appendix-A} for the proof)
\begin{lemma}\label{lem:corrbound}
Consider any LDGM code with bounded left and right degree. Consider two sets of information-bit nodes $A$, $B$ with bounded support. We have 
\begin{align}\label{KP-bound}
\big|\langle\prod_{a\in A}u_a\prod_{b\in B}u_b\rangle-\langle\prod_{a\in A}u_a\rangle\langle\prod_{b\in B}u_b\rangle\big|
\le 2\sum_{w\in
W_{AB}}\prod_{i\in w} \rho_i  
\end{align}
where $\rho_i=1$, if $i\in \bad$ and  $\rho_i=e^{4|l_i|}-1$, if $i\notin \bad$. 
\end{lemma}
\noindent The crucial feature of this lemma is that the $\rho_i$ are independent random variables because the walks are self-avoiding. Consequently, averaging over the noise realization in \eqref{KP-bound} 
\begin{align}\label{eq:avgcorrkp}
\bE_{\noise}\big|\langle\prod_{a\in A}u_a\prod_{b\in B}u_b\rangle-\langle\prod_{a\in A}u_a\rangle\langle\prod_{b\in B}u_b\rangle\big| 
&\le  2\sum_{w\in
W_{AB}}\prod_{i\in w} \bE[\rho_i]
\end{align}
Now,
\begin{align}
\bE[\rho_i]  & \leq  \bE[\rho_i\mid i\notin \bad] \mathbb{P}(i\notin\bad) + \bE[\rho_i\mid i\in \bad] \mathbb{P}(i\in\bad)
\cr &
\le  (e^{4H}-1)+\prob\big(|l|> H\big) = \delta(\epsilon, H)
\end{align}
For our class of channels we can choose $H=H(\epsilon)$  such that $K\delta(\epsilon, H(\epsilon))<1$.  
We get
\begin{align}
\bE_{\noise}\big|\langle\prod_{a\in A}u_a\prod_{b\in B}u_b\rangle - \langle\prod_{a\in_A}u_a\rangle\langle\prod_{b\in_B}u_b\rangle\big|
&
\le  2\sum_{w\in
W_{AB}} \delta(\epsilon)^{\vert w\vert}
\cr &
\le  2\vert A\vert \vert B\vert\sum_{d\geq {\rm dist}(A,B)} (K\delta(\epsilon, H(\epsilon)))^{d}
\cr &
\leq 
\frac{2|A||B|}{1-K\delta(\epsilon, H(\epsilon))}(K\delta(\epsilon))^{\dist A B}
\label{good-bound}
\end{align}
The second inequality is obtained by noticing that the number of selfavoiding random walks of length $\vert w\vert$ is certainly bounded by $K^{\vert w\vert}$. The factor $\vert A\vert \vert B\vert$ accounts for the maximum possible number of initial and final vertices. 
The correlation decay of the theorem is in fact a special case of this last bound for the choice $A=\partial i$ and $B=\partial j$.
\end{proof}

\vskip 0.25cm

We now look at GEXIT functions of the MAP and BP decoders.
For LDGM codes the functional giving the MAP-GEXIT function in \eqref{gexit} is 
\begin{align}\label{eq:GEXIT1}
\mathcal{G}(\langle x_i\rangle_0)=\frac{\Lambda'(1)}{P'(1)}  \int dl_i \frac{dc(l_i)}{d\epsilon} \bE_{l^{n\setminus i}}\ln\biggl\{  \frac{1+\langle x_i\rangle_0\tanh l_i}{1+\tanh l_i}\biggr\} 
\end{align}
Derivations of this formula can be found in \cite{Ruediger}, \cite{Macris2}.

The BP-GEXIT curve is given by the same functionals with $\langle x_i\rangle_0$ replaced by $\langle x_i\rangle_{0, d}^{BP}$. Consider $N_d(i)$ the neigborhood of node $i$, radius $d$ an even integer (all the vertices at graph-distance less or equal to $d$ from $i$). As is well known for an ensemble LDGM($\Lambda,P$) with bounded degrees, given $d$, if $n$ is large enough, the probability that $N_d(i)$ is a tree is $1-O(\frac{\gamma^d}{n})$ (where $\gamma$ depends only on $l_{\max}, k_{\max}$).
Thus when $d$ is fixed and $n\to+\infty$ the computational tree $T_{d}(i)$ and the neighborhood $N_d(i)$ match with high probability. This implies that 
\begin{equation}
\lim_{d\to+\infty}\lim_{n\to+\infty}g_{n,d}^{BP}(\epsilon) = 
\mathbb{E}_\code[\mathcal{G}(\langle x_i\rangle_{0, N_d(i)})\vert N_d(i)\,\text{is a tree}]
\end{equation}
where $\langle x_i\rangle_{0, N_d(i)}$ is the Gibbs bracket associated to the subgraph $N_d(i)$. The right hand side can exactly computed by performing the statistical mechanical sums on a tree and yields the DE formulas
\begin{align}\label{DE-gexit-ldgm}
\lim_{d\to+\infty}\lim_{n\to+\infty}g_{n,d}^{BP}(\epsilon)=\lim_{d\to \infty}\frac{\Lambda'(1)}{P'(1)}\int dl \frac{dc(l)}{d\epsilon} \bE_{\Delta^{(d)}}\ln\bigg\{\frac{1+\tanh \Delta^{(d)}\tanh l}{1+\tanh l}\bigg\} 
\end{align}
where both limits exist and
\begin{equation}
\tanh\Delta^{(d)} = \prod_{i=1}^{k}\tanh v_i^{(d)} 
\end{equation}
The $v_i^{(d)}$ are i.i.d random variables with distribution obtained from the iterative system of DE equations
\begin{align}
\eta^{(d)}(v) & =\sum_{l}\frac{l\Lambda_l}{\Lambda^\prime(1)} \int\prod_{i=1}^{l-1}du_i\,\ueta^{(d)}(u_i)\delta(v-\sum_{i=1}^{l-1}u_i)  \\
\ueta^{(d)}(u)&=\sum_{k}\frac{kP_k}{P^\prime(1)} \int dlc(l)\prod_{a=1}^{k-1}dv_a\,\eta^{(d-1)}(v_a) 
\delta(u-\tanh^{-1}(\tanh l \prod_{i=1}^{k-1}\tanh v_a) 
\end{align} 
with the initial condition $\eta^{(0)}(v)=\delta(v)$. It is well known that these equations are an iterative version of the replica fixed point equation \cite{Saad}.

\vskip 0.25cm

\begin{proof}[Proof of corollary \ref{corollary-1}, LDGM]
Expanding the logarithm in \eqref{eq:GEXIT1} and using Nishimori identities as in \cite{Macris2} we obtain  the expansion 
\begin{align}\label{gexitEXP}
\frac{\Lambda'(1)}{P'(1)} \sum_{p=1}^{+\infty}
&\frac{T_{2p}(\epsilon)
}{2p(2p-1)}\bigl(\bE_{\code,l^{n\setminus i}}[\langle x_i\rangle^{2p}_0]-1\bigr) \end{align}
where we recall that 
\begin{equation}
T_{2p}(\epsilon)=\frac{d}{d\epsilon}\int_{-\infty}^{+\infty} dl\, c(l)(\tanh l)^{2p}
\end{equation}
Note that in order to get the above expansion, it is important to use \eqref{eq:GEXIT1} as expressed here in terms of the {\it extrinsic} estimate.
Obviously, the series is absolutely convergent, uniformly with respect to $n$, 
for the class of channels $\mathcal{K}$. Thus by dominated convergence, the proof will be complete if we show that
\begin{equation}\label{right}
 \lim_{n\to+\infty}\bE_{\code,l^{n\setminus i}}[\langle x_i\rangle^{2p}_0]
=\lim_{d\to+\infty}\bE_{\Delta^{(d)}}[(\tanh\Delta^{(d)})^{2p}]
\end{equation}
Indeed one can then compute the $n\to+\infty$ limit term by term 
in 
\eqref{gexitEXP} and then resum the resulting series (which is again absolutely convergent, uniformly with respect to $d$)  to obtain 
\eqref{DE-gexit-ldgm}. 

Let us show \eqref{right}.
As pointed out before for $d$ fixed and $n\to+\infty$,
 $N_d(i)$ is a tree with high probability. Thus,
\begin{equation}\label{lim}
 \lim_{n\to\infty}\bE_{\code,l^{n\setminus i}}[\langle x_i\rangle^{2p}_0]
=\lim_{n\to\infty}\bE_{\code,l^{n\setminus i}}[\langle x_i\rangle^{2p}_0\vert N_d(i)\,\text{is a tree}]
\end{equation}
Notice that all paths connecting the bit $i$ with those outside 
$N_d(i)$ have a length at least equal to $d$, so because of Theorem 1 in the high noise regime $x_i$ is very weakly correlated to the complement of $N_d(i)$. Therefore we may expect that
\begin{align}\label{weakcorr}
\lim_{d\to+\infty}\lim_{n\to+\infty}\mathbb{E}_{\code, l^{n\setminus i}}[\vert\langle x_i\rangle_0^{2p} - \langle x_i\rangle_{0, N_d(i)}^{2p}\vert \mid N_d(i)\,\text{tree}] =0
\end{align}
Assuming for a moment that this is true we get from \eqref{lim},
\begin{equation}\label{limlim}
 \lim_{n\to\infty}\bE_{\code,l^{n\setminus i}}[\langle x_i\rangle^{2p}_0]
=\lim_{d\to\infty}\lim_{n\to\infty}\bE_{\code,l^{n\setminus i}}[\langle x_i\rangle^{2p}_{0, N_d(i)}\vert N_d(i)\,\text{is a tree}]
\end{equation}
and, when $N_d(i)$ is a tree, the Gibbs average $\langle x_i\rangle^{2p}_{0, N_d(i)}$ is explicitly computable and the right hand side of \eqref{limlim} reduces to
\begin{equation}\label{eq:BP}
\lim_{d\to +\infty}\mathbb{E}_{\Delta^{(d)}}[(\tanh \Delta^{(d)})^{2p}]
\end{equation}
This proves \eqref{right}.

Our task is now to prove \eqref{weakcorr}. 
Let $\mathring{N}_d(i)$ be the set of checks that are at distance $d$ from $i$. We order the checks $\in\mathring{N}_d(i)$ in a given (arbitrary) way, and call $\langle - \rangle_{0;\leq k}$ the Gibbs average with $l_k=0$ for the $k$ first checks of $\mathring{N}_d(i)$ (and $l_i=0$ for the root node). For the first one (call it $1$) we use 
$e^{l_{1}x_{1}}=\cosh l_{1}+x_{1}\sinh l_{1}$ to find
\begin{align}
\langle x_i\rangle_0=\langle x_i\rangle_{0;\leq 1}+\frac{\tanh l_{1} \big(\langle x_ix_{1}\rangle_{0;\leq 1}-\langle x_i\rangle_{0;\leq 1}\langle x_{1}\rangle_{0;\leq 1}\big)}{ 1+\langle x_{1}\rangle_{0;\leq 1}\tanh l_{1}}
\end{align}
Therefore
\begin{align}
\vert\langle x_i\rangle_0^{2p} - \langle x_i\rangle_{0;\leq 1}^{2p}\vert &
\leq 
2p\vert\langle x_i\rangle_0 - \langle x_i\rangle_{0;\leq 1}\vert
\cr &
\leq
2p\, t_1 \vert\langle x_ix_{1}\rangle_{0;\leq 1}-\langle x_i\rangle_{0;\leq 1}\langle x_{1}\rangle_{0;\leq 1}\vert
\end{align}
where 
\begin{equation}
t_k=\frac{\vert\tanh l_{k}\vert}{ 1-\vert\tanh l_{k}\vert}
\end{equation}
We can now take the second check of $\mathring{N}_d(i)$ (call it $2$) and show
\begin{align}
\vert\langle x_i\rangle_{0;\leq 1}^{2p} -  \langle x_i\rangle_{0;\leq 2}^{2p} \vert 
\leq
2p\, t_{2} \vert\langle x_ix_2\rangle_{0;\leq 2}-\langle x_i\rangle_{0;\leq 2}\langle x_{2}\rangle_{0;\leq 2}\vert
\end{align}
We can repeat this argument for all nodes of $\partial N_d(i)$ and use the triangle inequality to obtain
\begin{align}
\vert\langle x_i\rangle_{0}^{2p} -  \langle x_i\rangle_{0,N_d(i)}^{2p}\vert 
\leq
2p \sum_{k\in\mathring{N}_d(i)} t_k \vert\langle x_ix_k\rangle_{0;\leq k}-\langle x_i\rangle_{0;\leq k}\langle x_k\rangle_{0;\leq k}\vert
\end{align}
Indeed the Gibbs average with all $l_k=0$ for all $k\in \mathring{N}_d(i)$ is equal to $\langle x_i\rangle_{0,N_d(i)}$.
Now using the bound \eqref{good-bound} in the proof of theorem \ref{theorem-corr} for $K\delta(\epsilon)<1$, the last inequality implies
\begin{align}\label{finaly}
\mathbb{E}_{\code,l^{n\setminus i}}[\vert \langle x_i\rangle_0^{2p} -
\langle x_i\rangle_{0, N_d(i)}^{2p}\vert \mid N_d(i)\,\text{tree}]
\leq 
\frac{4p K^2\mathbb{E}[t]}{1-K\delta(\epsilon)}  K^d(K\delta(\epsilon) )^d
\end{align}
Note that for channels $\mathcal{K}$, for non-zero noise,
\begin{equation}
\mathbb{E}[t] =\mathbb{E}\biggl[\frac{\vert\tanh l\vert}{ 1-\vert\tanh l\vert}\biggr] \leq \mathbb{E}[e^{2\vert l\vert}]<\infty
\end{equation}
The right hand side of \eqref{finaly} does not depend on $n$, so it is immediate that \hfill
$\lim_{d\to+\infty}\lim_{n\to+\infty}$ vanishes as long as the noise is high enough such that $K^2\delta(\epsilon)<1$. This proves 
\eqref{weakcorr} and the corollary.
\end{proof} 
\vskip 0.25cm
To conclude, let us remark that, for the BIAWGNC the GEXIT formulas simplify considerably and there is a clear relationship to the magnetization,
\begin{align}\label{gexitAWGN}
\mathcal{G}(\langle x_i\rangle)&=\frac{1}{\epsilon^3}\frac{\Lambda'(1)}{P'(1)}\big(1-\bE_{\noise}[\langle x_i\rangle]\big)
\cr &
=
\frac{1}{\epsilon^3}\frac{\Lambda'(1)}{P'(1)}(1-
\bE_{l^n}[\tanh( l+ \tanh^{-1}\langle x_i\rangle_0)])
\end{align}
and
\begin{align}\label{gexit-gaussian}
\lim_{d\to+\infty}\lim_{n\to+\infty}g_{n,d}^{BP}(\epsilon)=\frac{1}{\epsilon^3}\frac{\Lambda'(1)}{P'(1)}(1-\bE_{l,\Delta^{(d)}}[\tanh (l+ \Delta^{(d)})])
\end{align}
The proof of corrolary 1 for BIAWGNC can thus proceed without expansions and is slightly simpler. The main ideas can be found in \cite{KuMa08itw} and we do not repeat them here. Note also that for the BEC there are similar simplifications that occur: this allows us to make a proof which avoids the second condition in the class of channels $\mathcal{K}$.

\section{\sc LDPC Codes: Low Noise}\label{section-4} 

In this section we prove theorem \ref{theorem-corr} and corollary \ref{corollary-1} for LDPC codes in a low noise regime.  As explained in section \ref{section-2} we first transform the problem to a dual one. The duality transformation reviewed here essentially is an application of Poisson's summation formula over commutative groups, and has been thoroughly discussed in the context of codes on graphs in \cite{Forney}. Here we need to know how the correlations transform under the duality, a point that does not seem to appear in the related literature.

\subsection{Duality formulas for the correlations}

Let $C$ be a binary parity check code and $C^{\perp}$ its dual. We apply the Poisson summation formula 
\begin{equation}
\sum_{x^n \in C} f(x^n) = \frac{1}{|C|}
\sum_{\tau^n\in C^{\perp}} \widehat{f}(\tau^n)
\end{equation}
where the Fourier (or Hadamard) transform is,
\begin{equation}
\widehat{f}(\tau^n) = \sum_{x^n \in \{-1,+1\}^n} f(x^n) e^{i\frac{\pi}{4} \sum_{j=1}^n (1-\tau_j)(1-x_j)}
\end{equation}
to the partition function $Z$ of an LDPC code $C$. The dual code $C^{\perp}$ is an LDGM with codewords given by
$\tau^n$ where 
\begin{equation}\label{codewords}
\tau_i=\prod_{a\in \partial i} u_a
\end{equation}
and $u_a$ are the $m$ information bits. A straigthforward application of the Poisson formula then yields the extended form of the MacWilliams identity,
\begin{equation}\label{duality}
 Z= \frac{1}{\vert C^\perp\vert}e^{\sum_{j=1}^n l_j}Z_{\perp}
\end{equation}
where 
\begin{equation}
Z_{\perp}  =  \sum_{u^m\in\{-1,+1\}^m} \prod_{i=1}^n (1+ e^{-2l_i}\prod_{a\in \partial i} u_a)
\end{equation}
This expression formaly looks like the partition function of an LDGM code with ``channel half-loglikelihoods'' $g_i$ such that $\tanh g_i = e^{-2l_i}$. This is truly the case only for the BEC($\epsilon$) where $l_i=0,+\infty$ and hence $g_i=+\infty,0$ which still correspond to a BEC($1-\epsilon$). The logarithm of partition functions is related to the input-output entropy and one recovers (taking the $\epsilon$ derivative) the well known duality relation between EXIT functions of a code and its dual on the BEC \cite{BrinkAshikmin}. For other channels however this is at best a formal (but still useful) analogy since the weights are negative for $l_i<0$ (and $g_i$ takes complex values). We introduce a bracket $\langle-\rangle_{\perp}$ which is not a true probabilistic expectation (but it is still linear)
\begin{equation}
\langle f \rangle_{\perp} = \frac{1}{Z_{\perp}}
\sum_{u^m\in\{-1,+1\}^m} f(u^m)\prod_{i=1}^n (1+ e^{-2l_i}\prod_{a\in i} u_a)
\end{equation}
The denominator may vanish, but it can be shown that when this happens the numerator also does so in a way that ensures the finiteness of the ratio (this becomes clear in subsequent calculations).
Taking logarithm of \eqref{duality} and then the derivative with respect to $l_i$ we find
\begin{align}\label{firstderivative}
\langle x_i \rangle
& = \frac{1}{\tanh 2l_i} -\frac{\langle\tau_i\rangle_{\perp}}{\sinh 2l_i} 
\end{align}
and differentiating once more with respect to  $l_j$, $j\neq i$ 
\begin{align}\label{secondderivative}
\langle x_i x_j\rangle - \langle x_i\rangle\langle x_j\rangle  =  \frac{\langle \tau_i\tau_j \rangle_{\perp} - \langle \tau_i \rangle_{\perp}\langle \tau_j \rangle_{\perp}}{\sinh 2l_i\sinh 2l_j}
\end{align}
We stress that in \eqref{firstderivative}, \eqref{secondderivative}, $\tau_i$ and $\tau_j$ are given by products of information bits \eqref{codewords}. The left hand side of \eqref{firstderivative} is obviously bounded. It is less obvious to see this directly on the right hand side and here we just note that the pole at $l_i=0$ is harmless since, for $l_i=0$, the bracket has all its ``weight`` on configurations with $\tau_i=1$. Similar remarks apply to \eqref{secondderivative}.
In any case, we will beat the poles by  using the following trick. For any $0<s<1$ and $\vert a\vert\leq 1$ we have $\vert a\vert\leq \vert a\vert^s$, thus
\begin{equation}\nonumber
\mathbb{E}_{l^n}[\vert
\langle x_ix_j\rangle - \langle x_i\rangle \langle x_j\rangle\vert] \leq 2^{1-s} \mathbb{E}_{l^n}[\vert
\langle x_i x_j\rangle - \langle x_i\rangle \langle x_j\rangle\vert^s]
\end{equation}
and using \eqref{secondderivative} and Cauchy-Schwarz
\begin{align}\label{spower}
\mathbb{E}_{l^n}[\vert
\langle x_ix_j\rangle - \langle x_i\rangle \langle x_j\rangle\vert]   \leq 2^{1-s}\mathbb{E}[(\sinh 2l)^{-2s}]
\mathbb{E}_{l^n}[\vert \langle \tau_i\tau_j \rangle_{\perp} - \langle \tau_i \rangle_{\perp}\langle \tau_j \rangle_{\perp}\vert^{2s}]^{1/2}
\end{align}
The prefactor is always finite for $0\leq s <\frac{1}{2}$ for our class of channels $\mathcal{K}$.  For example for the BIAWGNC we have
\begin{equation}\label{sinhbound}
\mathbb{E}[(\sinh 2l)^{-2s}]
\leq
\frac{c}{\vert 1-2s\vert}e^{-c^\prime\frac{s(1-2s)}{\epsilon^{2}}}
\end{equation}
for purely numerical constants $c, c^\prime>0$ and for the BSC we have
\begin{equation}
\mathbb{E}[(\sinh 2l)^{-2s}] \leq \biggl(\frac{2\epsilon(1-\epsilon)}{1-2\epsilon}\biggr)^{2s}
\end{equation}

\subsection{Decay of correlations for low noise}

We will prove the decay of correlations by applying a high temperature cluster expansion technique to $\mathbb{E}_{l^n}[\vert \langle \tau_i\tau_j \rangle_{\perp} - \langle \tau_i \rangle_{\perp}\langle \tau_j \rangle_{\perp}\vert^{2s}]$. As explained in section \ref{section-2} we need a technique that does not use the positivity of the Gibbs weights.
In appendix B we give a streamlined derivation of an adaptation of Berretti's 
expansion.
\begin{equation}\label{eq:berrettiexp}
\langle \tau_i\tau_j\rangle_{\perp}- \langle \tau_i\rangle_{\perp}\langle \tau_j\rangle_{\perp} 
= \frac12 \sum_{\hat{X}} K_{i,j}(\hat{X}) \Bigl(\frac{Z_{\perp}(\hat{X}^c)}{Z_{\perp}}\Bigr)^2 
\end{equation}
where 
\begin{equation}\label{eq:Kij}
K_{i,j}(\hat{X}) \equiv \sum_{\substack{u_a^{(1)}, u^{(2)}_a \\ a
\in \hat{X} }}\sum_{\substack{\Gamma ~\text{compatible} \\
\text{with} \hat{X}}} 
(\tau_i^{(1)}-\tau_i^{(2)})(\tau_j^{(1)}-\tau_j^{(2)}) \prod_{k\in \Gamma} E_k 
\end{equation}
and 
\begin{align}\label{eq:Ea}
E_k = \tau_k^{(1)} e^{-2 l_k} +\tau_k^{(2)} e^{-2 l_k} +\tau_k^{(1)}\tau_k^{(2)} e^{-4 l_k}
\end{align}
Here $u_a^{(1)}$ and $u_a^{(2)}$ are two independent copies of the information bits (these are also known as real
replicas) and $\tau_k^{(\alpha)}= \prod_{a\in k} u_a^{(\alpha)}$. To explain
what are $\hat{X}$ and $\Gamma$ we will refer to $a$-nodes (check nodes in the
Tanner graph representing the LDPC code) and $i$-nodes (variable nodes in the
Tanner graph representing the LDPC code). Given a subset $S$ of nodes of the
graph let $\partial S$ be the subset of neighboring nodes. 
In \eqref{eq:berrettiexp} the sum over $\hat{X}$ is carried over clusters of
$a$-nodes such that ``$\hat{X}$ is connected via hyperedges'': this means that a) $\hat{X}=\partial X$ for some connected subset $X$ of $i$-nodes; b) $X$ is connected if any pair of  $i$-nodes can be joined by a path all of whose variable nodes lie in $X$;  c) $\hat{X}$ contains both $\partial i$ and
$\partial j$. In the sum \eqref{eq:Kij} $\Gamma$ is a set of $i$-nodes (all
distinct). We say that ``$\Gamma$ 
is compatible with $\hat{X}$'' if: (i) $\partial\Gamma\cup
\partial i\cup\partial j = \hat{X}$, (ii) $\partial\Gamma \cap
\partial i \neq \phi$ and $\partial\Gamma \cap
\partial j \neq \phi$, (iii) there is a walk connecting $\partial i$ and $\partial j$ such that all its variable nodes are in $\Gamma$.
Finally,
\begin{equation}\label{eq:reducedpartitionfunction}
Z_{\perp}(\hat{X}^c)  = \sum_{\substack{u_a \\ a\in  \hat{X}^c}} 
\prod_{\substack{\text{all}\; i \;\text{s.t.} \\ \partial i \cap \hat{X} =
\phi}} (1 + e^{-2 l_i}\prod_{a\in i} u_a )
\end{equation}
The figure Fig. \ref{fig:berrettifigure} gives an example for all the sets appearing above.
\begin{figure}\label{fig:berrettifigure}
\begin{center}
\includegraphics[width=0.9\textwidth, height=0.7\textwidth]{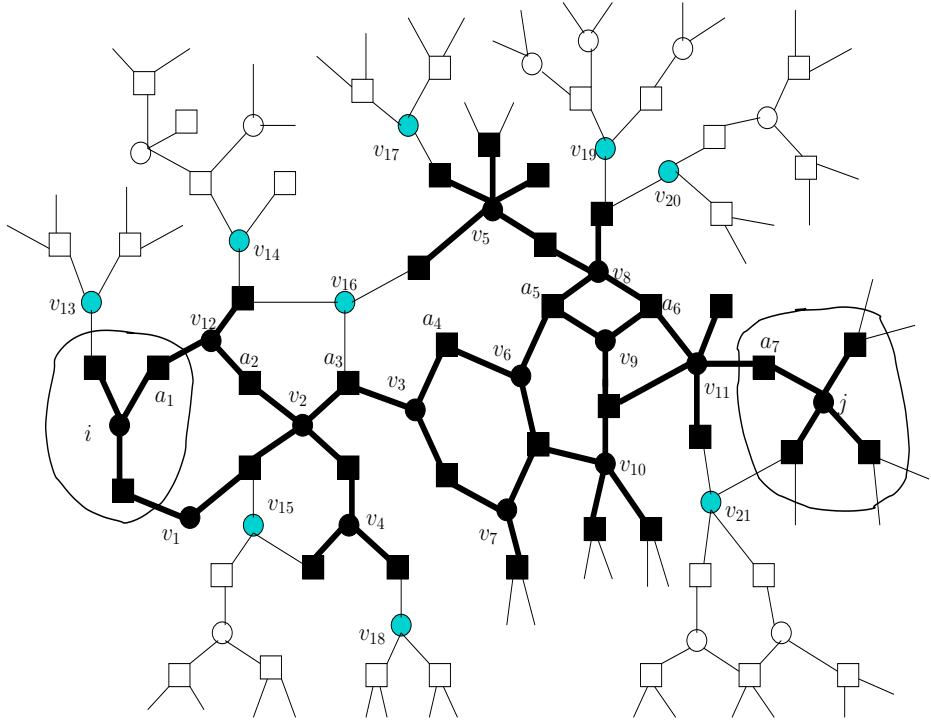}
\caption{In this figure we explain the various sets appearing in the cluster
expansion \eqref{eq:berrettiexp}. The Tanner graph represents the LDPC code with
variable nodes ($i$-nodes) denoted by circles and check nodes ($a$-nodes denoted
by squares). In this example the set $\hat{X}$ is the set of {\em dark} check nodes. It is easy
to verify that this choice of $\hat{X}$ satisfies all our conditions. Firstly,
let $X=\{i,j,v_1, v_2, v_3, v_4, v_5, v_6, v_7, v_8, v_9, v_{10}, v_{11},
v_{12}\}$ be a set of variable nodes (these are denoted by {\em dark} circles in
the figure). It is easy to check that the set of 
neighbours of $X$ is given by the {\em dark} check nodes which is $\hat{X}$.
Hence $\hat{X} = \partial X$. Secondly, any two variable nodes in $X$ are
connected by a path all of whose variable nodes lie in $X$, and thirdly,
$\hat{X}$ contains both $\partial i$ and $\partial j$. One choice for
$\Gamma = \{v_2, v_3, v_4, v_5, v_6, v_7, v_8, v_{10}, v_{11}, v_{12}\}$. It is
easy to check that $\Gamma$ is compatible with $\hat{X}$. The walk
$\{a_1v_{12}a_2v_{2}a_3v_{3}a_4v_{6}a_5v_{8}a_6v_{11}a_7\}$ connects $\partial
i$ and $\partial j$ and all its variable nodes lie in $\Gamma$. Another choice for $\Gamma$ would be the set $\{v_2, v_3, v_4, v_5, v_6, v_7, v_8, v_9, v_{10}, v_{11}, v_{12}\}$. In the
definition of $Z_{\perp}(\hat{X}^c), Z(\hat{X}^c)$ the {\em light} variable nodes, $v_{13},
v_{14}, v_{15}, v_{16}, v_{17}, v_{18}, v_{19}, v_{20}, v_{21}$, are not
present because they have a non-empty intersection with $\hat{X}$. }
\end{center}
\end{figure}

We are now ready to prove the theorem on decay of correlations.

\begin{proof}[Proof of theorem \ref{theorem-corr}, LDPC]
Because of \eqref{spower} it suffices to prove that \\ 
$\mathbb{E}_{l^n}[\vert \langle \tau_i\tau_j \rangle_\perp - \langle \tau_i \rangle_\perp\langle \tau_j \rangle_\perp\vert^{2s}]$
decays. 

The first step is to prove 
\begin{equation}\label{one-bound}
\biggl\vert\frac{Z_{\perp}(\hat{X}^c)}{Z_{\perp}}\biggr\vert \leq 1
\end{equation}
This ratio is not easily estimated directly because 
the weights in $Z_{\dcode}$ are not positive. However we can use the duality transformation \eqref{duality} backwards to get a new ratio of partition functions with positive weights,
\begin{equation}
\frac{Z_{\perp}(\hat{X}^c)}{Z_{\perp}} = \biggl(\exp{\sum_{\substack{\text{all}\,i\, \text{s.t} \\ \partial i \cap \hat{X} \neq\phi}} l_i}\biggr)\frac{\vert C^{\perp}(\hat{X}^c)\vert}{\vert C^{\perp}\vert} \frac{Z(\hat{X}^c)}{Z}
\end{equation}
with
\begin{equation}
Z(\hat{X}^c) = \sum_{\substack{x_i \\
\partial i \cap \hat{X} =\phi}} \prod_{\substack{\text{all}\,i\, \text{s.t} \\ \partial i \cap \hat{X} =\phi}} e^{l_ix_i} \prod_{a\in \hat{X}^c} \frac{1}{2}(1+\prod_{\substack{i\in a\,\text{and}\\ \partial i\cap\hat{X}=\phi}}x_i)
\end{equation}
This is the partition function corresponding to the subgraph induced by $a$-nodes of $\hat{X}^c$ and $i$-nodes such that $i\,\text{s.t}\,\partial i\cap\hat{X}=\phi$. Moreover $C^{\perp}(\hat{X}^c)$ is the dual of the later code $C(\hat{X}^c)$ defined on the subgraph. By standard properties of the rank of a matrix, the rank of the parity check matrix of $C(\hat{X}^c)$, which is obtained by removing rows (checks) and columns (variables) from the parity check matrix of $C$, is smaller than the rank of the parity check matrix of $C$. This implies $\vert C^{\perp}(\hat{X}^c)\vert \leq \vert C^{\perp}\vert$. Moreover 
\begin{equation}
\biggl(\exp{\sum_{\substack{\text{all}\,i\, \text{s.t} \\ \partial i \cap \hat{X} \neq\phi}} l_i}\biggr)
Z(\hat{X}^c) \leq Z
\end{equation}
To see this one must recognize that the left hand side is the sum of terms of $Z$ corresponding to $\sigma^n$ such that $\sigma_i=+1$ for $\partial i\cap\hat{X}\neq \phi$ (and all terms are $\geq 0$). These remarks imply 
\eqref{one-bound}.

Using 
$\vert\sum_{i}a_i\vert^{2s} \le \sum_{i} \vert a_i\vert^{2s}
$ for $0<2s<1$ and \eqref{one-bound} we find 
\begin{equation}\label{final}
\mathbb{E}_{l^n}[\vert \langle \tau_i\tau_j \rangle_\perp - \langle \tau_i \rangle_\perp\langle \tau_j \rangle_\perp\vert^{2s}] \leq 
\frac{1}{2^{2s}}\sum_{\hat{X}} \bE_{\noise}\bigl[|K_{i,j}(\hat{X})|^{2s}\bigr]  
\end{equation}
Trivially bounding the spins in $\eqref{eq:Ea}$ we deduce 
\begin{align}\nonumber
\bE_{\noise}\bigl[|K_{i,j}(\hat{X})|^{2s}\bigr]   & \leq 4^{\vert \hat X\vert}\sum_{\substack{\Gamma ~\text{compatible} \\
\text{with} \hat{X}}}(2^{2s}\bE[e^{-4sl}] + \bE[e^{-8sl}])^{\vert \Gamma\vert}
\\  &
\leq
 4^{\vert \hat X\vert}
\sum_{\substack{\Gamma ~\text{compatible} \\
\text{with} \hat{X}}} 2^{(4s+1)\vert\Gamma\vert}
\Delta(\epsilon)^{\vert \Gamma\vert}
\label{K1bound}
\end{align}
where 
\begin{equation}\label{Delta}
 \Delta(\epsilon)= 2^{2s}\bE[e^{-4sl}] + \bE[e^{-8sl}]
\end{equation}
Since $\Gamma$ is compatible with $\hat X$ we necessarily have
$\vert \partial\Gamma\vert \geq |\hat{X}| -
\vert\partial i\vert - \vert\partial j\vert$ and since $\vert\partial \Gamma\vert \leq \vert \Gamma\vert l_{\max}$, we get $\vert \Gamma\vert \geq 
(|\hat{X}| -2 l_{\max})/l_{\max}$. Also, the maximum
number of $a$-nodes which have an intersection with $\hat{X}$ is
$|\hat{X}|k_{\max}$. Thus there are at most $2^{|\hat{X}|k_{\max}}$ possible choices for
$\Gamma$. These remarks imply
\begin{equation}\label{K2bound}
\bE_{\noise}\bigl[|K_{i,j}(\hat{X})|^{2s}\bigr]  \le  2^{(2+k_{\max})|\hat{X}|} \Delta(\epsilon)^{(|\hat{X}| -2 l_{\max})/l_{\max}}
\end{equation}

From \eqref{final} and \eqref{K2bound} we get 
\begin{equation}
\mathbb{E}_{l^n}[\vert \langle \tau_i\tau_j \rangle_{\perp} - \langle \tau_i \rangle_{\perp}\langle \tau_j \rangle_{\perp}\vert^{2s}]\leq \frac{1}{2^{2s}}\sum_{\hat{X}}  2^{(2+k_{\max})|\hat{X}|} \Delta(\epsilon)^{(|\hat{X}| -2 \dlm)/2\dlm}
\end{equation}
The clusters $\hat{X}$ connect $\partial i$ and $\partial j$ and thus have sizes 
$\vert\hat{X}\vert\geq \frac{1}{2}\text{dist}(i,j)$. Moreover the number of clusters of a given size grows at most like $K^{\vert \hat{X}\vert}$ where $K=l_{\max}k_{max}$. Since for the class $\mathcal{K}$ we have for $s$ small enough, $\mathbb{E}[e^{sl}]\to 0$ as $\epsilon\to 0$ we can always chose $\epsilon$ small enough to make $\Delta(\epsilon)$ small enough and conclude the proof.
\end{proof}

\vskip 0.25cm


\subsection{Density evolution equals MAP for low noise}

In the case of LDPC codes the functional giving the MAP-GEXIT function in \eqref{gexit} is \cite{Macris}
\begin{align}\label{eq:GEXIT-ldpc}
\mathcal{G}(\langle x_i\rangle_0) =
= \int dl_i \frac{dc(l_i)}{d\epsilon} \bE_{l^{n\setminus i}}\ln\biggl\{  \frac{1+\langle x_i\rangle_0\tanh l_i}{1+\tanh l_i}\biggr\} 
\end{align}
Note that the only formal difference with the LDGM case is in the normalization factor; but of course now the Gibbs average pertains to the LDPC measure.
The BP-GEXIT curve is given by the same functionals with $\langle x_i\rangle_0$ replaced by the average on the computational tree $\langle x_i\rangle_{0, d}^{BP}$. As in section \ref{section-3} we introduce 
$N_d(i)$ the neigborhood of node $i$, radius $d$ an even integer. By the same arguments than in section \ref{section-3} we have again
\begin{equation}
\lim_{d\to+\infty}\lim_{n\to+\infty}g_{n,d}^{BP}(\epsilon) = 
\mathbb{E}_C[\mathcal{G}(\langle x_i\rangle_{0, N_d(i)})\vert N_d(i)\,\text{is a tree}]
\end{equation}
where $\langle x_i\rangle_{0, N_d(i)}$ is the Gibbs bracket associated to the graph $N_d(i)$. It is important to note that for $N_d(i)$ a tree the set of leaves $\mathring{N}_d(i)$ are variable nodes and have ``natural boundary conditions'' as given by the channel outputs. The statistical mechanical sums on a tree yield the DE formula
\begin{align}\label{DE-gexit}
\lim_{d\to+\infty}\lim_{n\to+\infty}g_{n,d}^{BP}(\epsilon)=\lim_{d\to \infty}\int dl \frac{dc(l)}{d\epsilon} \bE_{\Lambda^{(d)}}\ln\bigg\{\frac{1+\tanh\Lambda^{(d)}\tanh l}{1+\tanh l}\bigg\} 
\end{align}
where both limits exist and
\begin{equation}
\Lambda^{(d)} = \sum_{a=1}^k w_a^{(d)} 
\end{equation}
The $w_a^{(d)}$ are i.i.d random variables with distribution obtained from the iterative system of DE equations
\begin{align}
\zeta^{(d)}(w) & =\sum_{l}\frac{l\Lambda_l}{\Lambda^\prime(1)} \int\prod_{j=1}^{l-1}d\lambda_j\,\zeta^{(d)}(\lambda_j)\,\delta(w- \tanh^{-1}\bigl(\prod_{j=1}^{l-1}\tanh \lambda_j\bigr)) \nonumber \\
\widehat\zeta^{(d)}(\lambda)&=\sum_{k}\frac{kP_k}{P^\prime(1)} \int dl\,c(l)\prod_{a=1}^{k-1}dw_a\,\zeta^{(d-1)}(w_a)\, 
\delta(\lambda-l -\sum_{a=1}^{k-1} w_a^{(d)}) \nonumber
\end{align} 
with the initial condition $\eta^{(0)}(\lambda)=c(\lambda)$. As before, these equations are an iterative version of the replica fixed point equation \cite{Saad}. 
\begin{proof}[Proof of corollary 1, LDPC]
The first few steps are the same as in the proof for LDGM. First, we expand the logarithm in \eqref{eq:GEXIT-ldpc} and use Nishimori identities to obtain  a series expansion like \eqref{gexitEXP} (the prefactor $\frac{\Lambda'(1)}{P'(1)}$ is now absent). Second, we notice that since the resulting  series expansion is uniformly absolutely convergent it is enough to show that
\begin{equation}\label{newright}
 \lim_{n\to+\infty}\bE_{\code,l^{n\setminus i}}[\langle x_i\rangle^{2p}_0]
=\lim_{d\to+\infty}\bE_d[(\tanh\Lambda^{(d)})^{2p}]
\end{equation}
Thirdly, as before, one argues that this follows from
\begin{equation}\label{ssss}
\lim_{d\to+\infty}\lim_{n\to+\infty}\mathbb{E}_{C,l^{n\setminus i}}[\vert\langle x_i\rangle_0^{2p} -
\langle x_i\rangle_{0, N_d(i)}^{2p}\vert \mid N_d(i)\,\text{tree}] =0
\end{equation}
and because of $\vert b^{2p}- a^{2p}\vert\leq 2p \vert b-a\vert$ it is enough to show this for $2p$ replaced by $1$. Unfortunately one cannot proceed as simply as in the LDGM case: \eqref{ssss} is a consequence of the next two auxiliary lemmas stated below.
\end{proof}

Let $\langle - \rangle_{0; N_d(i)}^\infty$ be the bracket defined on the subgraph
$N_d(i)$ with $l_k=+\infty$ for $k\in \partial N_d(i)$. This in fact is formaly equivalent to fixing $x_k=+1$ boundary conditions on the leaves of the tree $k\in \mathring{N}_d(i)$.
The first lemma says that the bit estimate can be computed locally.
\begin{lemma}\label{first-lemma}
Under the same conditions than in corollary \ref{corollary-1},
\begin{equation}\label{72}
\lim_{d\to+\infty}\lim_{n\to+\infty}\mathbb{E}_{C,l^{n\setminus i}}[\vert \langle x_i\rangle_0 -
\langle x_i\rangle_{0, N_d(i)}^\infty\vert \mid N_d(i)\,\text{tree}] =0
\end{equation}
\end{lemma}
The second lemma says that at low enough noise free and $+1$ boundary conditions are equivalent
\begin{lemma}\label{second-lemma}
Under the same conditions than in corollary \ref{corollary-1},
\begin{equation}\label{73}
\lim_{d\to+\infty}\lim_{n\to+\infty}\mathbb{E}_{C,l^{n\setminus i}}[\langle x_i\rangle_{0, N_d(i)}^\infty - \langle x_i\rangle_{0, N_d(i)}
\vert \mid N_d(i)\,\text{tree}] =0
\end{equation}
\end{lemma}
We prove the first lemma. It will then be clear that the proof of the second one is essentialy the same except that the original full graph is replaced by $N_d(i)$, and thus it will be spared. 
\begin{proof}[Proof of lemma \ref{first-lemma}]
In \eqref{72} (and \eqref{73}) the root node $i$ has $l_i=0$ which turns out to be technically cumbersome because we really work in a low noise regime. For this reason we use
\begin{equation}
\langle x_i\rangle=\frac{\langle x_i\rangle_0 + \tanh l_i}{1+\langle x_i\rangle_0\tanh l_i},\qquad \langle x_i\rangle_{N_d(i)}^{\infty}=\frac{\langle x_i\rangle_{0,N_d(i)}^\infty + \tanh l_i}{1+\langle x_i\rangle_{0,N_d(i)}^\infty\tanh l_i}
\end{equation}
to deduce
\begin{equation}
\langle x_i\rangle_0 - \langle x_i\rangle_{0, N_d(i)}^\infty
= \frac{(1-(\tanh l_i)^2)(\langle x_i\rangle - \langle x_i\rangle_{N_d(i)}^\infty)}{(1-\langle x_i\rangle\tanh l_i)(1-\langle x_i\rangle_{N_d(i)}^\infty\tanh l_i)}
\end{equation}
This implies 
\begin{equation}
\vert\langle x_i\rangle_0 - \langle x_i\rangle_{0;N_d(i)}^\infty\vert
\leq \frac{1+\vert\tanh l_i\vert}{1-\vert \tanh l_i\vert}\,\vert\langle x_i\rangle - \langle x_i\rangle_{N_d(i)}^\infty\vert
\end{equation}
and averaging over the noise and using Cauchy-Schwarz,
\begin{align}
\mathbb{E}_{l^n\setminus i}[\vert\langle x_i\rangle_0 - \langle x_i\rangle_{0;N_d(i)}^\infty\vert] &
\leq 
2\mathbb{E}[e^{4\vert l\vert}]^{1/2}\mathbb{E}_{l^n}[\vert\langle x_i\rangle - \langle x_i\rangle_{N_d(i)}^\infty\vert^2]^{1/2}
\cr &
\leq 
2\sqrt{2}\,\mathbb{E}[e^{8\vert l\vert}]^{1/2}\,\mathbb{E}_{l^n}[\vert\langle x_i\rangle - \langle x_i\rangle_{N_d(i)}^\infty\vert]^{1/2}
\end{align}
Let us now prove
\begin{equation}
\lim_{d\to+\infty}\lim_{n\to+\infty}\mathbb{E}_{C,l^{n}}[\vert\langle x_i\rangle
- \langle x_i\rangle_{N_d(i)}^\infty\vert 
\mid N_d(i)\,\text{tree}]  = 0
\end{equation}
We order the variable nodes at the boundary $\mathring{N}_d(i)$ and consider the corresponding vector of loglikelihoods with components $\in \mathring{N}_d(i)$. If the first $k-1$ components of this vector are $l_1,...,l_{k-1}=+\infty$, the $k$-th component is $l_k^\prime$, and the other ones are i.i.d distributed as $c(l)$ (in other words they are ``natural``) we write  $\langle -\rangle_{\leq k-1}^{\infty}$. From the fundamental theorem of calculus, it is not difficult to see that
\begin{align}
\langle x_i\rangle - & \langle x_i\rangle_{N_d(i)}^\infty = - \sum_{k\in \mathring{N}_d(i)}
\int_{l_k}^{+\infty} dl_k^\prime\, \frac{d}{d l_k^\prime}\langle x_i\rangle_{\leq k-1}^{\infty}
\cr &
= -
\sum_{k\in \mathring{N}_d(i)}
\int_{l_k}^{+\infty} dl_k^\prime\, (\langle x_ix_k\rangle_{\leq k-1}^{\infty}
-\langle x_i\rangle_{\leq k-1}^{\infty}\langle x_k\rangle_{\leq k-1}^{\infty})
\end{align}
Using $\vert a\vert\leq \vert a\vert^s$ for any $0<s<1$ and $\vert a\vert \leq 1$ we get
\begin{align}
\vert\langle x_i\rangle - & \langle x_i\rangle_{N_d(i)}^\infty\vert 
\cr &
\leq 
2^{1-s}\sum_{k\in \mathring{N}_d(i)}
\int_{l_k}^{+\infty} dl_k^\prime\, \vert\langle x_ix_k\rangle_{\leq k-1}^{\infty}
-\langle x_i\rangle_{\leq k-1}^{\infty}\langle x_k\rangle_{\leq k-1}^{\infty}\vert^s
\end{align}
Let $\langle -\rangle_{\leq k-1}^{\infty,\perp}$ be the dual bracket (with the first $k$ components of $\mathring{N}_d(i)$
$l_1=...=l_{k-1}=+\infty$ and the $k$-th component equal to $l_k^\prime$). Because of \eqref{secondderivative} we have
\begin{align}
\vert\langle x_i\rangle - & \langle x_i\rangle_{N_d(i)}^\infty\vert 
\cr &
\leq 
2^{1-s}\sum_{k\in N_d(i)}
\int_{l_k}^{+\infty} dl_k^\prime\, \frac{\vert\langle \tau_i\tau_k\rangle_{\leq k-1}^{\infty,\perp}
-\langle \tau_i\rangle_{\leq k-1}^{\infty,\perp}\langle \tau_k\rangle_{\leq k-1}^{\infty,\perp}\vert^s}{(\sinh 2l_i\sinh 2l_k^\prime)^{2s}}
\end{align}
Note that the denominator in the integral is important to make the integral convergent for $l_k^\prime\to \infty$. Moreover at $l_i$ and $l_k^\prime=0$ is harmless as long as for $2s<1$. 
The next step is to use the cluster expansion in order to estimate
\begin{align}
\mathbb{E}_{l^n}\biggl[\int_{l_k}^{+\infty} dl_k^\prime\,\frac{\vert\langle \tau_i\tau_k\rangle_{\leq k-1}^{\infty,\perp}
-\langle \tau_i\rangle_{\leq k-1}^{\infty,\perp}\langle \tau_k\rangle_{\leq k-1}^{\infty,\perp}\vert^s}{(\sinh 2l_i\sinh 2l_k^\prime)^{2s}}\biggr]
\end{align}
By following similar steps than in the proof of theorem \ref{theorem-corr} one obtains an upper bound similar to \eqref{K2bound} except that the likelihoods of the end points are weighted differently and therefore there are two factors of $\Delta(\epsilon)$ (see \eqref{Delta}) replaced by 
\begin{equation}
\mathbb{E}\biggl[\frac{2^{2s}e^{-4sl} + e^{-8sl}}{(\sinh 2l)^{-2s}}\biggr] <\infty
\qquad{\rm and}\qquad
\mathbb{E}\biggl[\int_{l}^{+\infty} dl^\prime\, \frac{2^{2s}e^{-4sl^\prime} + e^{-8sl^\prime}}{(\sinh 2
l^\prime)^{-2s}}\biggr]<\infty
\end{equation}
Finally we can average over the code ensemble conditional on the event
that $N_d(i)$ is a tree. Since the clusters $\hat X$ that connect $\partial i$ and $\partial k$, $k\in N_d(i)$ have size $\vert \hat X\vert k_{\max}$ we obtain the result as long as $\Delta(\epsilon)$ is small enough, for $\epsilon$ small enough.
\end{proof}

\section{Large block length versus large number of iterations}\label{section-5}

In the LDGM case we prove the exchange of limits $d,n\to+\infty$ for the BSC channel. As will become clear one needs the decay of correlations (or covariance) of the Gibbs measure on the computational tree for $d\gg n$. Hence the likelihoods are not independent r.v: the proof of theorem \ref{theorem-corr} still goes through in the case
of the BSC. The only difference is that in lemma \ref{lem:corrbound} we can take $H> \frac{1}{2}\ln\frac{1-\epsilon}{\epsilon}$ such that $\bad=\emptyset$ and $\rho_{\pi(j)}= e^{4\vert l_{\pi(j)}\vert} - 1 = \frac{4\vert 1-2\epsilon\vert}{(1-\vert 1-2\epsilon\vert)^2}$ for all $j\in T_d(i)$. 
\begin{lemma}[Decay of correlations for the BP decoder, LDGM on BSC]\label{lemma-BP}
Consider communication with a fixed general LDGM code with blocklength size $n$ and bounded degrees $l_{\max}$, $k_{\max}$, over the BSC($\epsilon$). We can find $c>0$, a small enough numerical constant such that for $l_{\max} k_{\max}\vert 1- 2\epsilon\vert <c$ we have, for any given realization of the channel outputs, 
\begin{equation}
\vert \langle x_i x_j\rangle_{d}^{BP} - \langle x_i \rangle_{d}^{BP} \langle x_j\rangle_{d}^{BP}\vert
\leq c_1e^{-c_2(\epsilon){\rm dist}(i,j)}
\end{equation}
where $i$ is the root of the computational tree, $j$ and arbitrary node, 
 $c_1>0$ a numerical constant and $c_2(\epsilon)>0$ depending only on $\epsilon$, $l_{\max}$, $k_{\max}$. Moreover $c_2(\epsilon)$ increases like $\ln\vert 1- 2\epsilon\vert$ as $\epsilon\to \frac{1}{2}$.
\end{lemma}
Basically, this result is contained in \cite{Tatikonda} where it is obtained by Dobrushin's criterion. Note that it is valid for fixed noise realizations and not only on average. The unbounded case would require to take averages but then, on the computational tree one has to control moments $\mathbb{E}[\rho_\pi(i)^m]$ and this requires more work. The following proof is a simple application of this lemma.
\begin{proof}[Proof of theorem 2, LDGM, BSC]
We take for the number of iterations of the BP decoder $d\gg n$. On the computational tree $T_d(i)$ we consider the subtree of root $i$ and depth $d^\prime\ll n$. This subtree is a smaller computational tree
$T_{d^\prime}(i)\subset T_{d}(i)$ and $d^\prime\ll n\ll d$. Let 
$\mathring{T}_{d^\prime}(i)$ the leaves $k$ with ${\rm dist}(i,k)=d^\prime$ and order them in an arbitrary way. Consider the Gibbs measure $\langle - \rangle_{d; \leq k}^{BP}$ where for the first $k$ checks of $\mathring{T}_{d^\prime}(i)$ we set $l_{\pi(k)}=0$ in \eqref{BP-ldgm}. Proceeding as in section \ref{section-3} we get
\begin{equation}
\vert \langle x_i\rangle_{d}^{BP} - \langle x_i\rangle_{d^\prime}^{BP}\vert
\leq 
\sum_{k\in \mathring{T}_{d^\prime}(i)} t_{\pi(k)} \vert \langle x_i x_k\rangle_{d; \leq k}^{BP}
- \langle x_i \rangle_{d; \leq k}^{BP} \langle x_k\rangle_{d; \leq k}^{BP} \vert
\end{equation}
For the BSC, $t_{\pi(k)} = \frac{\vert 1-2\epsilon\vert}{1 - \vert 1- 2\epsilon \vert}$. From 
lemma \ref{lemma-BP} for $\vert 1-2\epsilon \vert$ small enough (but independent of $n,d$)
\begin{equation}
\vert \langle x_i\rangle_{d}^{BP} - \langle x_i\rangle_{d^\prime}^{BP}\vert = O( K^{d^\prime} e^{-c_2(\epsilon) d^\prime})
\end{equation}
In this equation $O(-)$ is uniformly bounded with respect to $n$ and $d$ (and the noise realizations of course). Recall the GEXIT function of the BP decoder
\begin{equation}
g_{n,d}(\epsilon) = \frac{\Lambda'(1)}{P'(1)}  \int dl_i \frac{dc(l_i)}{d\epsilon} \bE_{\code, l^{n\setminus i}}\ln\biggl\{  \frac{1+\langle x_i\rangle_{0,d}^{BP}\tanh l_i}{1+\tanh l_i}\biggr\} 
\end{equation}
Since for $\vert 1-2\epsilon \vert \ll 1$, $\vert\tanh l_i \vert= \frac{1}{2}\vert \ln \frac{1-\epsilon}{\epsilon}\vert\ll 1$, one can easily show
\begin{equation}\label{88}
g_{n,d}(\epsilon) = g_{n,d^\prime}(\epsilon) + O( K^{d^\prime} e^{-c_2(\epsilon) d^\prime})
\end{equation}
For example one could proceed by expanding the $\ln$ in powers of $\vert\tanh l_i\vert$ and estimate the series term by term.
Now since $O(-)$ is uniformly bounded with respect to $n,d$ \eqref{88} implies for $d^\prime$ fixed
\begin{equation}
\lim_{n\to +\infty}\liminf_{d\to \infty}g_{n,d}(\epsilon) = \lim_{n\to +\infty}g_{n,d^\prime}(\epsilon) + O( K^{d^\prime} e^{-c_2(\epsilon) d^\prime})
\end{equation} 
Now we take the limit $d^\prime\to +\infty$,
\begin{equation}
\lim_{n\to +\infty}\liminf_{d\to \infty}g_{n,d}(\epsilon) = \lim_{d^\prime\to +\infty}\lim_{n\to +\infty}g_{n,d^\prime}(\epsilon)
\end{equation}
A similar result with $\limsup$ replacing $\liminf$ is derived in the same way. 
\end{proof}

\section{\sc Conclusion}

In this paper we have shown that cluster expansion techniques of statistical mechanics are a valuable tool for the theory of error correcting codes on graphs.
We have not investigated the regimes of high noise for LDPC codes and low noise for LDGM codes. In the case of LDPC codes and high noise we are able to prove decay of correlations for ensembles that contain a sufficient fraction of degree one variable nodes. Indeed one can eliminate the degree one nodes and convert the problem to a new graphical model containing a mixture of hard parity check constraints and soft LDGM type weigths. If the density of soft weights is high enough the analysis of the present paper can be extended (see \cite{KuMa07TC} for a summary). Combining theses ideas with duality one may also treat special ensembles of LDGM codes for low noise. This approach however is not entirely satisfactory and it is not clear how to directly go about with cluster expansions in these regimes. 

We hope that the ideas and techniques investigated in the present work could
have other applications in coding theory and more broadly random graphical
models. Let us mention that various forms of correlation decay have been
investigated recently for the random $K$-SAT problem at low constraint density,
by different methods \cite{MoShah}. This has allowed the authors to prove that
the replica symmetric solution is exact at low constraint density.  In \cite{Chertkov}
the authors derive a new type of expansion called ``loop expansion'' in an
attempt to compute corrections to BP equations. The link to traditioanl cluster
expansions is unclear to us, and also it would  be interesting to develop
rigorous methods to control the loop expansions.  Finaly, we would also like to
point out the work \cite{Procacci} where a new derivation of the Gilbert-Varshamov
bound is presented using the Mayer expansion for a hard-sphere systems.

\section*{Acknowledgment} 
The work of Shrinivas Kudekar has been supported by a grant of the Swiss National Foundation 200020-113412.

\appendix

\section{\sc Cluster Expansions}\label{Appendix-A}
In this appendix we explain the derivation of the two cluster expansions that we use. In the statistical mechanics literature these have been derived for spin systems with pair interactions on regular graphs. It turns out that they can be adapted to our setting. We try to give a self-contained by still reasonably short derivation here.

\subsection{Cluster expansion for LDGM codes}
Here we adapt the cluster expansion of Dreifus-Klein-Perez in \cite{Klein}. In the process we prove lemma \ref{lem:corrbound} stated in section \ref{section-3}. It will be very convenient to use the following compact notation
\begin{equation}
\prod_{a\in X} u_a = u_X,\qquad \text{for any set}\,\,\, X\subset \{1,...,m\}
\end{equation}
In particular the code-bits $x_i = \prod_{a\in \partial i} u_a$ become $u_{\partial i}$, $i=1,...,n$ and the correlation of lemma \ref{lem:corrbound} becomes 
\begin{equation}
\langle u_A u_B\rangle - \langle u_A\rangle \langle u_B\rangle
\end{equation}
It is first necessary to rewrite the Gibbs measure \eqref{gibbs-ldgm} in a form such that the exponent is positive
\begin{equation}
\frac{1}{Z}\prod_{i=1}^m e^{l_ix_i} = \frac{1}{Z^\prime}\prod_{i=1}^m e^{l_i u_{\partial i} + \vert l_i\vert}
\end{equation}
where $Z^\prime$ is the appropriately modified partition function.
We introduce the replicated measure, which is the product of two copies,
\begin{equation}
\frac{1}{Z^{\prime 2}}\prod_{i=1}^m e^{l_i (u_{\partial i}^{(1)}+u^{(2)}_{\partial i}) + 2\vert l_i\vert} 
\end{equation}
Thus we now have two replicas of the information-bits $u^{(1)}_1,...,u^{(1)}_m$ 
and $u^{(2)}_1,...,u^{(2)}_m$. The Gibbs bracket for the replicated measure is denoted by $\langle - \rangle_{12}$. It is easy to see that 
\begin{align}
\langle u_A u_B\rangle-\langle u_A\rangle\langle u_B\rangle = \frac{1}{2}\langle(u^{(1)}_A-u^{(2)}_A)(u^{(1)}_B-u^{(2)}_B )\rangle_{12}
\end{align}
Recall that $\bad=\{i\mid \vert l_i\vert >H\}$ for some fixed number $H$, and set 
\begin{equation}
e^{l_i(u_{\partial i}^{(1)}+u_{\partial i}^{(2)})+2|l_i|} - 1=K_i
\end{equation}
It will be important to keep in mind later that $K_i\geq 0$.
We have 
\begin{align}
\frac{1}{2}\langle(u^{(1)}_A-u^{(2)}_A)&(u^{(1)}_B-u^{(2)}_B )\rangle_{12} 
\nonumber \\ 
&=\frac{1}{2Z'^2}\sum_{u^{(1)},u^{(2)}}f_Af_B\prod_{i\in \bad}e^{l_i(u_{\partial i}^{(1)}+u_{\partial i}^{(2)})+2|l_i|}\prod_{i\in \bad^c} (1+K_i) 
\nonumber \\ 
&=\frac{1}{2Z'^2}\sum_{u^{(1)},u^{(2)}}f_Af_B\prod_{i\in \bad}e^{l_i(u_{\partial i}^{(1)}+u_{\partial i}^{(2)})+2|l_i|}\sum_{G\subseteq \bad^c}\prod_{i\in G} K_i
\nonumber \\
& = \frac{1}{2Z'^2}\sum_{G\subseteq \bad^c}\sum_{u^{(1)},u^{(2)}}f_Af_B\prod_{i\in \bad}e^{l_i(u_{\partial i}^{(1)}+u_{\partial i}^{(2)})+2|l_i|}\prod_{i\in G} K_i
\label{the-sum-over-G}
\end{align}
where $f_X = u^{(1)}_X-u^{(2)}_X$, $X=A,B$.

Take a term with given $G\subseteq\bad^c$ in the last sum. We say that "$G$ connects $A$ and $B$" if and only if there exist a self-avoiding walk\footnote{See section \ref{section-3} for the definition of these walks.} $w_{ab}$ with initial variable node $a\in A$, final variable node  $b\in B$ and such that all check nodes of $w_{ab}$ are in $G\cup\bad$\footnote{Note that it is really $G\cup\bad$ that connects $A$ and $B$. Since $\bad$ is fixed our definition is valid}. The crucial point is that: if a set $G$ does not connect $A$ and $B$, then it gives a vanishing contribution to the sum. We defer the proof of this fact to the end of this section. For the moment let us show that it implies the bound in lemma
\ref{lem:corrbound}. The positivity of $K_i$ implies 
\begin{align}
\vert\langle u_A&u_B\rangle  - \langle u_A\rangle\langle u_B\rangle\vert
\cr &
\le \frac{2}{Z'^2}\sum_{\substack{G\subseteq\bad^c\\ G \,\,\text{connects A and B} }}\sum_{u^{(1)},u^{(2)}}\prod_{i\in \bad}e^{l_i(u_{\partial i}^{(1)}+u_{\partial i}^{(2)})+2|l_i|}\prod_{i\in G} K_i 
\cr &
\leq 
\frac{2}{Z'^2}
\sum_{w\in
W_{AB}} \sum_{G'\subseteq
\bad^c\setminus w}\sum_{u^{(1)},u^{(2)}}\prod_{i\in
\bad}e^{l_i(u_{\partial i}^{(1)}+u_{\partial i}^{(2)})+2|l_i|}\prod_{i\in w\setminus\bad} \mathfrak{h}_i\prod_{i\in G'} K_i
\end{align}
In the second inequality we used $K_i\leq e^{4|l_i|} - 1\equiv \mathfrak{h}_i$. 
Now resumming over $G^\prime\subseteq\bad^c\setminus w$ we obtain
\begin{align}\label{eq:expansion4}
\vert\langle u_A&u_B\rangle-\langle u_A\rangle\langle u_B\rangle\vert 
\cr &
\le 
\frac{2}{Z'^2}\sum_{w\in
W_{AB}}\prod_{i\in w\setminus\bad} \mathfrak{h}_i\sum_{u^{(1)},u^{(2)}}\prod_{i\in \bad}e^{l_i(u_{\partial i}^{(1)}+u_{\partial i}^{(2)})+2|h_i|}\prod_{i\in\bad^c\setminus w } (1+K_i) 
\nonumber \\ 
&
\le 
\frac{2}{Z'^2}\sum_{w\in
W_{AB}}\prod_{i\in w\setminus\bad} \mathfrak{h}_i\sum_{u^{(1)},u^{(2)}}\prod_{i\in \bad}e^{l_i(u_{\partial i}^{(1)}+u_{\partial i}^{(2)})+2|h_i|}\prod_{i\in w\setminus\bad}(1+K_i) \prod_{i\in\bad^c\setminus w } (1+K_i) \nonumber \\ 
&
= 
2\sum_{w\in
W_{AB}}\prod_{i\in w\setminus\bad} \mathfrak{h}_i  
\end{align}
The second inequality follows by inserting extra terms $1+K_i\geq 1$ for $i\in w\setminus\bad$, and the second by reconstituting $Z^{\prime 2}$ in the numerator. Now, the last line is equal to 
\begin{equation}
2\sum_{w\in
W_{AB}}\prod_{i\in w} \rho_i , \,\,\,\,\, \rho_i=1, i\in \bad\,\,\,\, {\rm and}\,\,\,\,  \rho_i=\mathfrak{h}_i, i\notin \bad
\end{equation}
Hence the bound \eqref{KP-bound}.

It remains to explain why, if $G$ does not connect $A$ and $B$, the $G$-term does not contribute to \eqref{the-sum-over-G}. Let $\partial G\cup\partial\bad$ be the set of variable nodes connected to the check nodes $G\cup\bad$. We define a partition $\partial G\cup\partial\bad = V_A\cup V_C\cup V_B$ into three sets of variable nodes. $V_A$ is the set of all variable nodes $v$ such that there exist a self-avoiding walk $w_{av}$ connecting some $a\in A$ to $v$, and such that all ckeck nodes of  $w_{av}$ are in $G\cup\bad$. $V_B$ is similarly defined with $B$ and $b\in B$ instead of $A$. Finally $V_C= (\partial G\cup\partial\bad) \setminus
(V_A\cup V_B)$. By construction $V_C\cap V_A=V_C\cap V_B=\emptyset$. The point is that if $G$ does not connect $A$ and $B$, then $V_A\cap V_B=\emptyset$. Indeed, otherwise there would be a $u\in V_A\cap V_B$ with a walk $w_{au}$ and a walk $w_{ub}$ both with all check nodes in $G\cup \bad$, but this would mean that $G$ connects $A$ and $B$ through the walk $w_{au}\cup w_{ub}$. 
We also define three sets of check nodes 
$C_{A}= (G\cup\bad)\cap\partial V_{A}$, $C_{B}= (G\cup\bad)\cap\partial V_{B}$ and $C_C= (G\cup\bad)\setminus (G_A\cup G_B)$. Again the three sets are disjoint when $G$ does not connect $A$ and $B$: indeed if there exists $c\in C_A\cap C_B$ then $c$ belongs to both $V_A$ and $V_B$ which we just argued is impossible. This situation is depicted on figure (3).
\begin{figure}\label{figure-walks}
\begin{center}
\includegraphics[width=0.8\textwidth, height=0.7\textwidth]{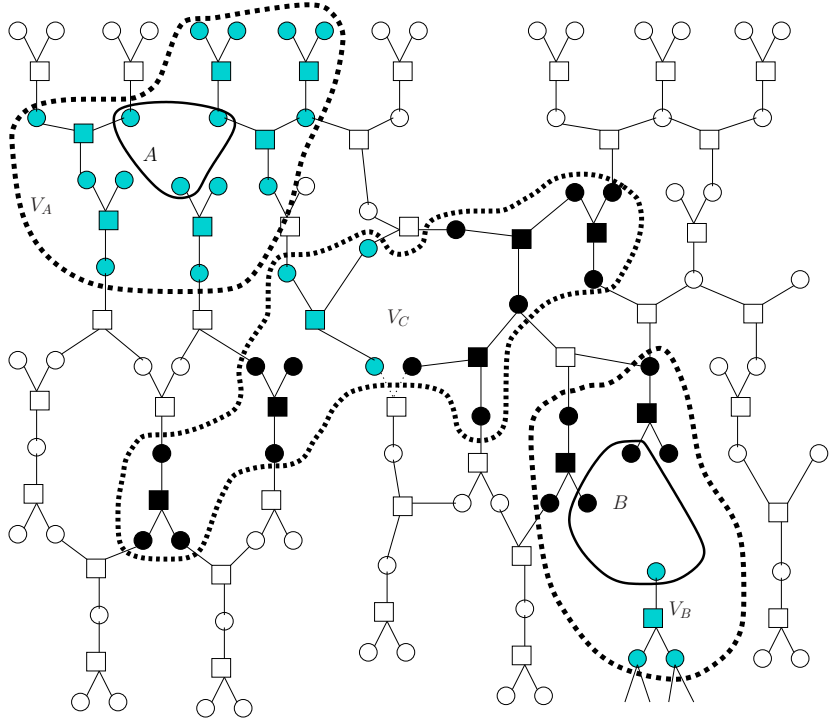}
\caption{On the LDGM graph, $\bad$ is depicted by the dark squares. The set $G\subseteq \bad^c$ is depicted by the light squares. $A$ and $B$ contain both three nodes; there does not 
exist a self-avoiding walk that connects these two sets with all its check nodes in $G$. The sets of variable nodes $V_A$, $V_B$ and $V_C$ are disjoint as well as the sets of check nodes $G_A$, $G_B$ and $G_C$: these sets are enclosed in the dotted areas.} 
\end{center}
\end{figure}

Now we examine a term of \eqref{the-sum-over-G} for a $G$ that does not connect $A$ and $B$. Expanding the product $f_Af_B$, using linearity of the bracket and symmetry under exchange of replicas
$(1)\leftrightarrow (2)$, it is equal to the difference $I -II$ where
\begin{equation}
I=\frac{1}{2Z'^2}\sum_{u^{(1)},u^{(2)}}u^{(1)}_Au^{(1)}_B\prod_{i\in \bad}e^{l_i(u_{\partial i}^{(1)}+u_{\partial i}^{(2)})+2|l_i|}\prod_{i\in G} K_i
\end{equation}
\begin{equation}
II=\frac{1}{2Z'^2}\sum_{u^{(1)},u^{(2)}}u^{(1)}_Au^{(2)}_B\prod_{i\in \bad}e^{l_i(u_{\partial i}^{(1)}+u_{\partial i}^{(2)})+2|l_i|}\prod_{i\in G} K_i
\end{equation}
Because of the disjointness of the sets $V_{A,B,C}$ and $C_{A,B,C}$ (the areas enclosed in dotted lines, see figure (3)) one can, in $I$ and $II$, factor the sums $\sum_{u^{(1)},u^{(2)}}$ in a product of three terms (in fact there is a fourth trivial term which is a power of $2$ coming from the bits outside the dotted areas). Then by symmetry $1\leftrightarrow 2$ one recognizes that $I=II$. Thus $I-II$ and this proves that $G$ does not contribute to \eqref{the-sum-over-G} when it does not connect $A$ and $B$.

\subsection{Cluster expansion for LDPC codes}

Here we adapt the Berretti cluster expansion to our setting. For more details we refer to \cite{Berretti}, \cite{Frohlich}. 
Consider the {\it replicated} partition function 
\begin{align}
Z_{\perp}^2 & = \sum_{u^{(1)}, u^{(2)}\in\{-1,+1\}^m} \prod_{k=1}^n (1 + \tau^{(1)}_k e^{-2 l_k})(1 + \tau^{(2)}_k e^{-2 l_k})
\end{align}
here $u^{(1)}=u^{(1)}_1, \dots, u^{(1)}_m;  u^{(2)}=u^{(2)}_1, \dots, u^{(2)}_m$
are two replicas of the information bits and $\tau^{(1)}_k = \prod_{a\in k}
u^{(1)}_a, \tau^{(2)}_k = \prod_{a\in k} u^{(2)}_a$. 
We have
\begin{align}
\langle \tau_i\tau_j \rangle_{\perp} - \langle \tau_i \rangle_{\perp}\langle \tau_j \rangle_{\perp} & = 
\frac12\langle (\tau^{(1)}_i - \tau^{(2)}_i)(\tau^{(1)}_j - \tau^{(2)}_j)\rangle_{\perp,12}
\end{align}
where $\langle \cdot \rangle_{\perp,12}$ corresponds to the replicated system. 
We denote $f_{i}= \tau^{(1)}_i - \tau^{(2)}_i$, $f_j=\tau^{(1)}_j - \tau^{(2)}_j$.
Then we have
\begin{align}
\langle f_{i}f_j\rangle_{\perp,12} = \frac{1}{Z_{\perp}^2}\sum_{u^{(1)}, u^{(2)}}&
f_{i}f_j \prod_{k}(1 + E_k)
\end{align}
where $E_i$ is defined in \eqref{eq:Ea}.
Expanding the product we get,
\begin{align}\label{eq:fijberretti}
\langle f_{i}f_j\rangle_{\perp,12} & = \frac{1}{Z_{\perp}^2} \sum_{u^{(1)}, u^{(2)}} f_{i}f_j \sum_{V\subset \spins} \prod_{k\in V} E_k \nonumber \\
& = \frac{1}{Z_{\perp}^2} \sum_{V\subset \spins}\sum_{u^{(1)}, u^{(2)}} f_{i} f_j\prod_{k\in V} E_k
\end{align}
where $\spins$ denotes the set of all variable nodes of the original Tanner graph for the LDPC code and $V$ is any subset of distinct variable nodes. 
Suppose  
$V\subset \spins$ is such that {\em one cannot create a walk} (i.e. on the original Tannger graph of the LDPC code, a set of
alternating variable and check nodes) connecting any check node in
 $\partial i$, to any check node in 
$\partial j$, and which has all its variable nodes contained 
entirely in $V$. Then we can partition $V$ into three mutually disjoint sets of variable nodes,
$V_1, V_2, V_3$ such that $V_1 \ni i$, $V_2 \ni j$ and $V_3=V\setminus (V_1\cup
V_2)$. Note also that $\partial V_1, \partial V_2, \partial V_3$ are mutually
disjoint otherwise we can create a walk between $\partial i$ and $\partial j$. Thus we can write  
\begin{align}
\sum_{u^{(1)}, u^{(2)}}f_if_j\prod_{k\in V} E_k  = \sum_{\substack{u^{(1)}, u^{(2)} \\ u^{(1)}_a, u^{(2)}_a
\in \partial V_1}}
f_i\prod_{k\in V_1} E_k 
&
\sum_{\substack{u^{(1)}, u^{(2)} \\ u^{(1)}_a, u^{(2)}_a
\in \partial V_2}}f_j \prod_{k\in V_2} E_k  
\nonumber \\
&
\times\sum_{\substack{u^{(1)}, u^{(2)} \\ u^{(1)}_a, u^{(2)}_a
\in \partial V_3}}\prod_{k\in V_3} E_k
\label{vanish}
\end{align} 
This implies that \eqref{vanish} vanishes. 
This is seen by using the antisymmetry of
$f_i$ (or $f_j$) and the symmetry of $E_k$, 
 under the exchange $(1)\leftrightarrow (2)$.
Thus only those $V$  which contain a walk with all its variable nodes in $V$ and
which intersects both $\partial i$ and $\partial j$ contributes to the sum in
\eqref{eq:fijberretti}.  

For any given $V$ (contributing to the sum) we construct the set of variable nodes $\Gamma_V$ as follows. 
$\Gamma_V$ is the union of all  maximal connected clusters of 
distinct variable nodes in $V$, such that each of those connected clusters intersects 
$\partial i\cup \partial j$.  Let $\Gamma_V^c = V \setminus \Gamma_V$. Clearly,
there exists such a set because we know that the walk which connects $\partial
i$ and $\partial j$ is a subset of $\Gamma_V$. Let
$\hat{X}_V = \partial \Gamma_V \cup \partial i \cup \partial j$ be a set of
check nodes. It is not difficult to see that $\hat{X}_V$ satisfies all the
requirements of the set $\hat{X}$ in the sum \eqref{eq:berrettiexp}. Indeed,
consider $X_V = \Gamma_V \cup i \cup j$. By construction $\partial X_V =
\hat{X}_V$; any two variable nodes in $X_V$ are connected by a walk with all its
variable nodes in $X_V$; $\hat{X}_V$ contains both $\partial i$ and $\partial
j$. Also note that $\Gamma_V$ is compatible with $\hat{X}_V$ as is required in
the sum \eqref{eq:Kij}. Indeed, by construction $\partial \Gamma_V\cup \partial
i \cup \partial j = \hat{X}_V$; $\partial \Gamma_V\cap \partial i \neq \phi$ and 
$\partial \Gamma_V\cap \partial j \neq \phi$; there exists a walk between
$\partial i$ and $\partial j$ with all its variable nodes in $\Gamma_V$.

With this we can write
\begin{align}
\langle f_{i}f_j\rangle_{\perp,12}
& = \frac{1}{Z_{\perp}^2} \sum_{V\subset \spins}\Big\{\sum_{\substack{u_a^{(1)},
u^{(2)}_a
\\ a \in \partial \Gamma_V\cup \partial i\cup\partial j}} f_{i}f_j \prod_{k\in \Gamma_V} E_k \Big\}
\Big\{ \sum_{\substack{u_a^{(1)}, u_a^{(2)} \\ \text{remaining}\; a}}\prod_{k\in
\Gamma_V^c} E_k\Big\} \nonumber \\
& = \frac{1}{Z_{\perp}^2} \sum_{\hat{X}}\sum_{\substack{V\subset \spins\; : \\ \hat{X}_V = \hat{X}}}
\Big\{\sum_{\substack{u_a^{(1)},u^{(2)}_a
\\ a \in \hat{X}}} f_{i}f_j \prod_{k\in \Gamma_V} E_k \Big\}
\Big\{ \sum_{\substack{u_a^{(1)}, u_a^{(2)} \\ \text{remaining}\; a}}\prod_{k\in
\Gamma_V^c} E_k\Big\} 
\end{align}
Now we resum over the sets $V$ such that $\hat X_V=\hat X$. These consist of $\Gamma$ compatible with $\hat X$ and the rest $\mathcal G$ which does not intersect $\hat X$. So
\begin{align}
\langle f_{i}f_j\rangle_{\perp,12} 
& = \frac{1}{Z_{\perp}^2} \sum_{\hat{X}}\Bigg\{\sum_{\substack{u^{(1)}_a,
u^{(2)}_a \\ a
\in \hat{X} }}\sum_{\substack{\Gamma \text{compatible} \\
\text{with} \hat{X}}} 
f_{i}f_j \prod_{k\in \Gamma} E_k \Bigg\}\Bigg\{
\sum_{\substack{u^{(1)}_a, u^{(2)}_a \\ a\in \hat{X}^c}}
\sum_{\substack{\mathcal{G} \subseteq \spins \\
\partial \mathcal{G}\cap \hat{X}=\phi}}\prod_{k\in \mathcal{G}} E_k\Bigg\} \nonumber \\
& = \frac{1}{Z_{\perp}^2} \sum_{\hat{X}}\Bigg\{\sum_{\substack{u^{(1)}_a,
u^{(2)}_a \\ a
\in \hat{X} }}\sum_{\substack{\Gamma \text{compatible} \\
\text{with} \hat{X}}} 
f_{i} f_j\prod_{k\in \Gamma} E_k \Bigg\}\Bigg\{
\sum_{\substack{u^{(1)}_a, u^{(2)}_a \\ a\in \hat{X}^c}}
\prod_{\substack{\text{all}\; k\;\text{s.t.} \\ \partial k \cap \hat{X} = \phi }}(1+E_k) \Bigg\}
\end{align} 
The last bracket is equal to \eqref{eq:reducedpartitionfunction} and we recognize Berretti's expansion. Figure (4) shows a sample set $V$ and
$\Gamma_V$ which give a non-vanishing contribution. 
\begin{figure}\label{fig:berrettifigure2}
\begin{center}
\includegraphics[width=0.85\textwidth,height=0.65\textwidth]{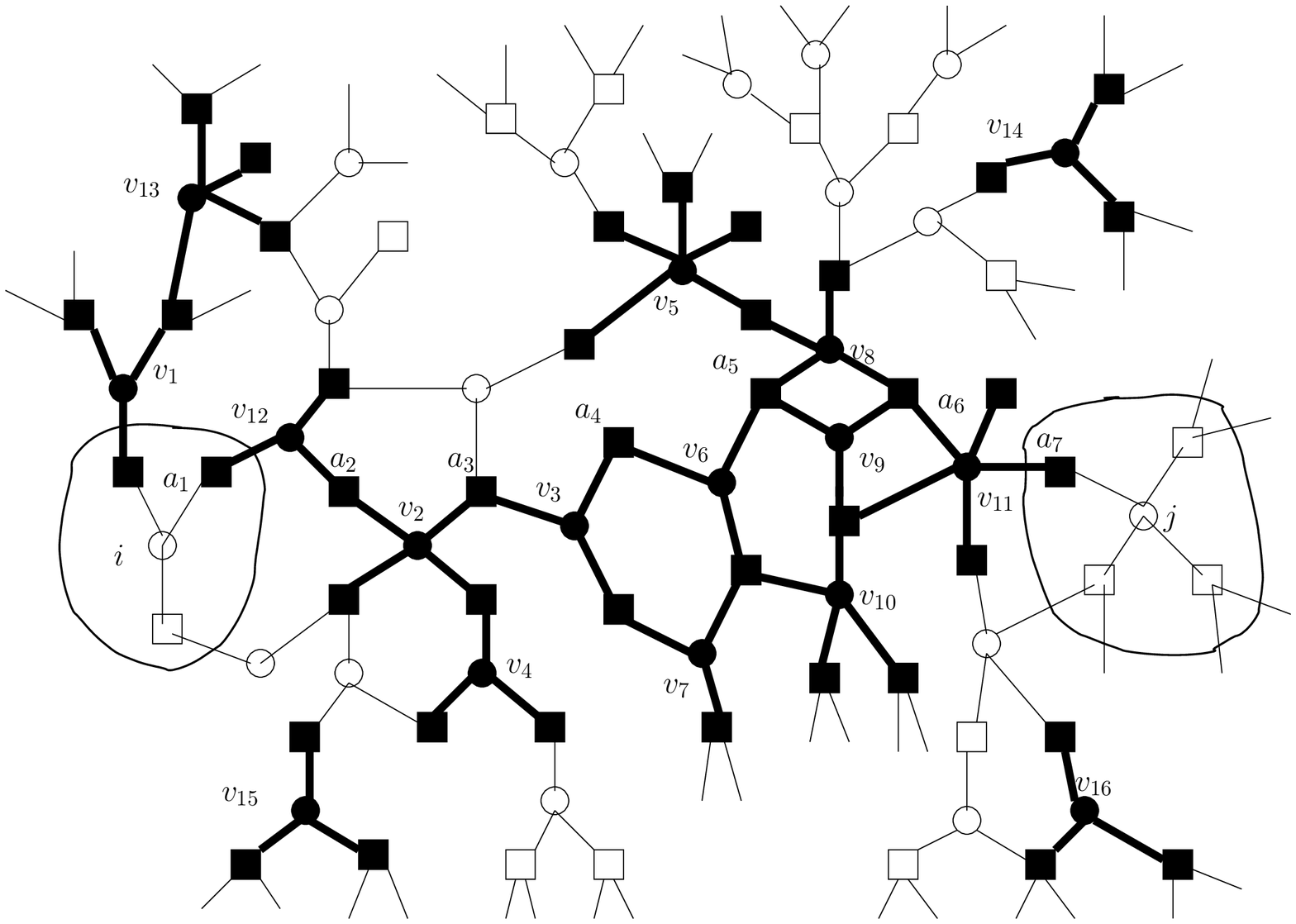}
\caption{The {\em dark} variable nodes form the set $V=\{v_1, \dots, v_{16}\}$. The walk $a_1v_{12}a_2v_2a_3v_3a_4v_6a_5v_9a_6v_{11}a_7$ connects $\partial i$ to $\partial j$ and hence this $V$ has a non-vanishing contribution. 
$\Gamma_V=\{v_{1}, \dots, v_{12}\}$, is union of the two maximal connected clusters
$\{v_1, v_{13}\}$ and $\{v_{12}, v_2, v_3, v_4, v_5, v_6, v_7, v_8, v_9, v_{10},
v_{11}\}$ which has  
intersection with $\partial i \cup \partial j$. $\Gamma_V^c = \{v_{14},
v_{15}, v_{16}\}$.}
\end{center}
\end{figure}

\end{document}